\documentclass[11pt,a4paper]{article}

\usepackage[english]{babel}

\usepackage{geometry}
\usepackage{graphicx}
\usepackage{xspace}
\usepackage{amssymb,amsmath,amsthm}
\usepackage{marvosym}
\usepackage{subcaption}
\usepackage{delarray}
\usepackage{enumitem}
\usepackage{tikz}
\usepackage{hyperref}
\graphicspath{{figures/}}
\usepackage{ifthen}
\usepackage{authblk}

\newtheorem{theorem}{\textbf{Theorem}}
\newtheorem{lemma}{\textbf{Lemma}}
\newtheorem{proposition}{\textbf{Proposition}}
\newtheorem{corollary}{\textbf{Corollary}}
\newtheorem{myremark}{\textbf{Remark}}

\newtheorem{definition}{\textbf{Definition}}

\newtheorem{remark*}{Remark}

\newtheorem{conjecture}{\textbf{Conjecture}}
\newtheorem{open}{\textbf{Open question}}
\newcommand{\set}[1]{\ensuremath{\bgroup\left\{#1\right\}\egroup}\xspace}
\newcommand{\ie}{\textit{i.e.}\@\xspace}
\newcommand{\wrt}{\textit{w.r.t.}\@\xspace}
\newcommand{\etc}{\textit{etc.}\@\xspace}
\newcommand{\N}{\mathbb{N}}
\newcommand{\Z}{\mathbb{Z}}

\newcommand{\B}{\mathbb{B}}

\renewcommand{\O}{\mathcal{O}}
\renewcommand{\o}{o}
\newcommand{\NP}{{\mathsf{NP}}}
\newcommand{\Poly}{{\mathsf{P}}}
\newcommand{\NC}{{\mathsf{NC}}}
\newcommand{\AC}{{\mathsf{AC}}}
\newcommand{\TC}{{\mathsf{TC}}}
\renewcommand{\L}{{\mathsf{L}}}
\newcommand{\NL}{{\mathsf{NL}}}
\newcommand{\dom}{{\mathsf{dom}}}
\newcommand{\spanv}{{\mathsf{span}^+}}
\newcommand{\red}[1]{\leq_{#1}}
\newcommand{\CVP}{\ensuremath{\mathsf{CVP}}\xspace}
\newcommand{\MCVP}{\ensuremath{\mathsf{MCVP}}\xspace}
\newcommand{\PCVP}{\ensuremath{\mathsf{PCVP}}\xspace}
\newcommand{\MPCVP}{\ensuremath{\mathsf{MPCVP}}\xspace}
\newcommand{\LOGCFL}{\ensuremath{\mathsf{LOGCFL}}\xspace}

\newcommand{\neighborhood}{\mathcal{N}}
\newcommand{\neighborhoodM}{\mathcal{N}_\text{M}}
\newcommand{\neighborhoodVN}{\mathcal{N}_\text{VN}}
\newcommand{\cfgs}{\N^{\Z^d}}
\renewcommand{\H}[1]{\ensuremath{\mathsf{H}\left(#1\right)}\xspace}
\newcommand{\distribution}{\mathcal{D}}
\newcommand{\structure}[1]{\ensuremath{\left\langle#1\right\rangle}\xspace}
\newcommand{\indic}{{\bf 1}}
\newcommand{\degree}{d}

\newcommand{\distributionone}{\mathcal{D}_1}
\newcommand{\odometer}{{\mathsf{odo}}}
\newcommand{\odometerseq}{{\mathsf{odo}_{\toseq}}}
\newcommand{\PRED}{\ensuremath{\mathsf{PRED}}\xspace}
\newcommand{\SPRED}{\ensuremath{\mathsf{S\text{-}PRED}}\xspace}
\newcommand{\FSPRED}{\ensuremath{\mathsf{1^\text{st}col\text{-}S\text{-}PRED}}\xspace}
\newcommand{\CPRED}{\ensuremath{\mathsf{Compute\text{-}PRED}}\xspace}
\newcommand{\toseq}{\rightharpoonup}
\newcommand{\donne}{\to}
\newcommand{\sdonne}{\stackrel{*}{\donne}}
\newcommand{\donnes}[1]{\stackrel{#1}{\rightharpoonup}}
\newcommand{\sdonnes}{\stackrel{*}{\rightharpoonup}}

\title{How hard is it to predict sandpiles on lattices? A survey.}
\author[1]{Enrico Formenti}
\author[2]{K\'evin Perrot}
\affil[1]{Université Côte d’Azur, CNRS, I3S, France.}
\affil[2]{Aix-Marseille Univ., Toulon Univ., CNRS, LIS, Marseille, France.}
\date{}

\begin{document}
\setlist[itemize,enumerate]{nosep}
\maketitle

\begin{abstract}
  Since their introduction in the 80s, sandpile models have raised
  interest for their simple definition and their surprising dynamical
  properties. In this survey we focus on the computational complexity 
  of the prediction problem, namely, the complexity of knowing, 
  given a finite configuration $c$ and a cell $x$ in $c$, if cell $x$
  will eventually become unstable. This is an attempt to
  formalize the intuitive notion of ``behavioral complexity'' that one
  easily observes in simulations. However, despite many efforts and nice
  results, the original question remains open: how hard is it to
  predict the two-dimensional sandpile model of Bak, Tang and Wiesenfeld?
\end{abstract}

\section{Introduction}

Langton proposed to describe complex dynamical systems as being at the
``edge of chaos'' \cite{Langton1990}. Complexity arises in a context that is
neither too ordered, \ie not exhibiting a rigid structure allowing to efficiently understand
and predict the future state of the system, nor completely chaotic, \ie avoiding
pseudo-random behaviour and uncomputable long-term effects. 
From a computer science point of view, complex dynamical systems are an object of great
interest because they precisely model physical systems that are able to perform
non-trivial computation.

In 1990, Moore \emph{et. al.}\xspace started to formalize the intuitive notion of
``complexity'' of a system, through the computational complexity of predicting
the behaviour of the system 
\cite{Griffeath1996LifeWD,Machta1996,Machta1997,Machta1993,moore1997predicting,woods110}
(for other kinds of complexity in dynamical systems, see for instance~\cite{Formenti2013,formenti2017,dennunzio2008,formenti2007b,Dennunzio2012,CattaneoDFP09,CattaneoCDFM14,DennunzioFP13}). 
Computational complexity theory is actually a perfect fit to capture the complexity of systems
able to compute. In turned out that the hierarchy of complexity classes offers
a very precise way to characterize the behavioural complexity of discrete
dynamical systems. For example if a system has a $\Poly$-hard prediction
problem, it means that it is able to efficiently simulate a general purpose
sequential computer (such as a Turing machine), whereas if the prediction
problem is much below in the hierarchy, let say in $\L$, then the system can
only compute under severe space restriction, and therefore cannot perform
efficiently any computation if we assume $\Poly \neq \L$. In this precise sense
the former would be more complex than the latter.

At the same time, the sandpile model of Bak, Tang and Wiesenfeld~\cite{1987-BakTangWiesenfeld-SOC,bak88}
gained interest. It exhibits
both a ``complex'' behaviour and a very elegant algebraic 
structure~\cite{1990-Dhar-SelfOrganizedSandpileAutomatonModels}.
Unsurprisingly, sandpile models are capable of universal 
computation~\cite{1996-GolesMargenstern-SPMUniversal}. In 1999, Moore and 
Nilsson began to apply the computational complexity vocabulary to capture the 
intuitive ``complexity'' of sandpile models~\cite{1999-Moore-complexitySandpiles}. Moreover, they
observed a dimension sensitivity that received great 
attention. It is the purpose of this 
survey to review such very interesting results and to generalise some of them.

Sandpile models are a subclass of number-conserving cellular automata where we
are given a $d$-dimensional lattice ($\Z^d$) with a finite amount of sand
grains at each cell. A local rule applied in parallel at every cell let grains
topple: if the sand content at a cell is greater or equal to $2d$,
then the cell gives one grain to each of the $2d$ cells it touches (two cells
in each dimension). This is the very first sandpile model of Bak, Tang and
Wiesenfeld, which they defined for $d=2$. This is also the sandpile model
studied by Moore and Nilsson, for which they proved the following foundations.
\begin{itemize}
  \item In dimension one it is possible to predict efficiently the dynamics
    with a parallel algorithm (complexity class $\NC$).
  \item In dimension three or more it is not possible to predict efficiently
    the dynamics with a parallel algorithm, unless some classical complexity
    conjecture is wrong (unless $\Poly = \NC$, since prediction is proven to
    be a $\Poly$-hard problem). In other terms, in this case
    the dynamics is inherently sequential.
\end{itemize}

This survey concentrates on lattice $\Z^d$, because of this interest in the dimension
sensitivity. In a more general setting than lattices, sandpiles can very
easily embed arbitrary computation and become almost always hard to predict from
a computational complexity point of view~\cite{1997-Goles-UniversalityCFG}.

Sandpile models have close relatives, the family of \emph{majority} cellular automata. 
Indeed, though the latter model is not number-conserving, open questions on its 
two-dimensional prediction are remarkably similar \cite{Moore1997}. Goles \emph{et. al.}
made progress in various directions to capture the essence of $\Poly$-completeness
in majority cellular automata \cite{gmmo17,gm14,gm16,gmpt17,gmt13}, with 
notable applications of $\NC$ algorithms from \cite{jaja92}. Cellular automata with finite
support are very close to sandpiles when considered under the sequential update
policy. However, in this context we witness a general increase of the complexity
of the (decidable) questions about the dynamics which seems not to happen in 
sandpiles~\cite{formenti2017}.
\medskip

The paper is structured as follows.
In Section~\ref{s:def} we define sandpile models on lattices with
uniform neighborhood, formulate three versions of the prediction problem,
introduce some classical considerations, and briefly review the
complexity classes at stake. Subsequent sections survey known
results, and generalise some of them or propose conjectures.
All prediction problems are in $\Poly$ (Section \ref{s:poly}).
The dimension sensitivity for arbitrary sandpile models
generalizes as follows: in dimension one prediction is in $\NC$
(Section \ref{s:1d}), in dimension three or above it is
$\Poly$-complete (Section \ref{s:3d}), and in two dimensions the
precise complexity classification remains open for the original
sandpile model with von Neumann neighborhood, though insightful
results have been obtained around this question (Section
\ref{s:2d}). Finally, we briefly mention how undecidability may
arise when the finiteness condition on the initial configuration
is relaxed (Section~\ref{s:undecidable}).

\section{Definitions}
\label{s:def}

Since their introduction by Bak, Tang and Wiesenfeld in~\cite{1987-BakTangWiesenfeld-SOC}, sandpiles
underwent many generalizations. In this survey we propose a
general framework which tries to cover all such models.
However, we will focus only over lattices of 
arbitrary dimension and uniform (in space and time) \textit{number-conserving} local rules. 
This section includes formalization of folklore terminology and considerations extended to this 
general setting.

\subsection{Sandpile models on lattices with uniform neighborhood}
Let $\N_+$ (resp. $\N_-$) denote the set of strictly positive (resp. negative) integers.
For any \textit{dimension} $d\in\N_+$, a \textit{cell} is a point in $\Z^d$. A
\textit{configuration} is an assignment of a finite number of sand grains to each cell \ie it is an element 
of $\N^{\Z^d}$. A \textit{sandpile model} is a structure \structure{\neighborhood, \distribution, \theta}
where $\neighborhood$ is a finite subset of $\Z^d \setminus \{0^d\}$ called \textit{neighborhood} 
and $\distribution \in {\N_+}^\neighborhood$ is \textit{distribution}
of sand grains \wrt the neighborhood $\neighborhood$ (it is required that $\dom(\distribution)=\neighborhood$)
and
$\theta=\sum_{x\in\neighborhood} \distribution(x)$ is
the \textit{stability threshold}. To avoid irrelevant technicalities, we will consider only \textit{complete} neighborhoods $\neighborhood$, that is to say
such that
\begin{equation}\label{eq:span}
  \spanv(\neighborhood)=\Z^d,
\end{equation}
where $\spanv(\neighborhood)$ is the set of positive integer linear combinations of cell coordinates from $\neighborhood$.
Moreover, remark that, for simplicity sake, we assumed that $0^d\notin\neighborhood$ \ie cells do not belong to their own neighborhoods. 
Indeed, allowing  $0^d\in\neighborhood$ would only correspond to having irremovable grains in each cell.

The dynamics associated with a sandpile model is the
parallel application of the following local rule:
if a cell has at least $\theta$ grains,
then it redistributes $\theta$ of its grains to {\em its neighborhood} $x+\neighborhood$,
according to the distribution $\distribution$.
More formally, denote $F: \N^{\Z^d} \to \N^{\Z^d}$ the \textit{global rule} which associates any configuration $c\in\N^{\Z^d}$ with
a configuration $F(c)\in\N^{\Z^d}$ defined as follows
\begin{equation}
  \label{eq:sandpile}
  \forall x \in \Z^d :
  \big(F(c)\big)(x)=
  c(x)
  - \theta\H{c(x)-\theta}
  + \sum_{y \in\neighborhood} \distribution(y)\H{c(x+y)-\theta}
\end{equation}
where $\H{n}$ equals $1$ if $n \geq 0$, and equals $0$ otherwise (the classical Heaviside function). From the Equation~\eqref{eq:sandpile}, it is
clear that the knowledge of the distribution $\distribution$
suffices to completely specify the dynamics since from the domain of $\distribution$ one can
deduce $\neighborhood$ and the dimension $d$, and from $\distribution$ one finds $\theta$.
However, we shall prefer to provide explicitly $\structure{\neighborhood, \distribution, \theta}$, at least
in this introductory material.

The system characterized by Equation~\eqref{eq:sandpile} is \textit{number-conserving} in the sense
the total number of sand grains is conserved along its evolution as stated in the following.

\begin{proposition}[Number conservation]
  \label{prop:mass}
  For any configuration $c \in \N_+^{\Z^d}$, it holds
  \[\#(F(c))=\#(c)\]
  where $\#(c)=\sum_{x \in \Z^d} c(x)$. Note that $\#(c)$ could be positive infinity.
\end{proposition}

Figure~\ref{fig:sandpile} provides an illustration of
a simple sandpile model and its dynamics.

\begin{figure}[t]
  \centering
  \raisebox{.5em}{\includegraphics[scale=.8]{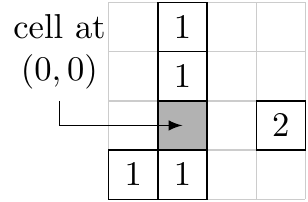}}
  \hspace*{1cm}
  \includegraphics[scale=.8]{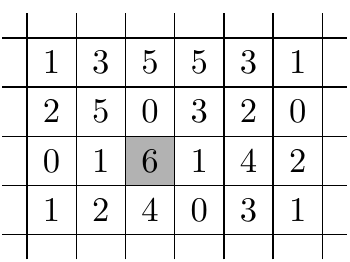}
  \raisebox{2em}{$\overset{F}{\mapsto}$}
  \includegraphics[scale=.8]{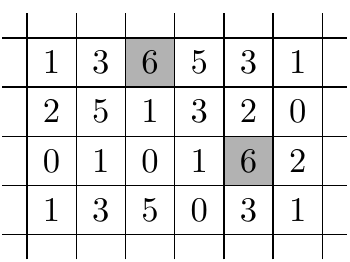}
  \raisebox{2em}{$\overset{F}{\mapsto}$}
  \includegraphics[scale=.8]{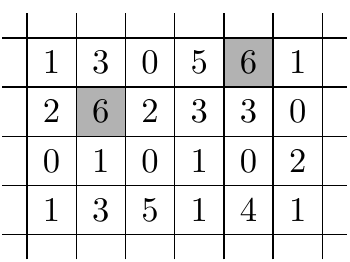}
  \caption{An example of sandpile model neighborhood and
  distribution (left, $\theta=6$), and two steps of the global
  rule from a finite configuration (right, outside the pictured
  region the configuration is considered as initially empty).
  Unstable cells are marked. Only the interesting portion of the
  configuration is drawn.}
  \label{fig:sandpile}
\end{figure}


\begin{myremark}
  \label{remark:vn}
  In this general framework, the \textit{Moore sandpile model} of dimension
  $d$ and radius $r$ corresponds to $\structure{\neighborhoodM,\distributionone,(2r+1)^d-1}$
  where $\neighborhoodM=\set{-r,\ldots,r}^d \setminus \set{0^d}$
  and $\distributionone$ is the constant function equal to $1$ for any element of
  its domain. The \textit{von Neumann sandpile model} of dimension $d$ and radius
  $r$ corresponds to
  \structure{\neighborhoodVN, \distributionone, 2rd} where
  \[
    \neighborhoodVN=\set{(x_1,\ldots,x_d)\in\set{-r,\ldots,r}^d\,\mid\,\exists i : (x_i \neq 0 \text{ and } \forall j \neq i:x_j=0)}
    \enspace
  \]
  (this is the original model of Bak, Tang and Wiesenfeld \cite{1987-BakTangWiesenfeld-SOC} for $d=2$ and $r=1$).
\end{myremark}

Denote $c \donne c'$ whenever $c'=F(c)$, and let $\sdonne$
be the reflexive and transitive closure of $\donne$. Given a cell $x$, remark that its neighbors might play
different roles. Indeed, one can distinguish the \textit{out-neighbors} of $x$ as the set of
cells $x+\neighborhood = \{ y \mid (y-x) \in \neighborhood \}$ from the
\textit{in-neighbors} which are $x-\neighborhood = \{ y \mid (x-y) \in \neighborhood \}$. 

\begin{myremark}
  Another even more general way to define sandpile models is on an arbitrary
  multi-digraph
  $G=(V,A)$ (where $A$ is a multiset), where each vertex has finite in-degree and finite out-degree. In this case, configurations $c$ are taken in $\N^V$ and the local
  rule would be: if $v \in V$ contains at least $\degree^+(v)$ grains ($\degree^+(v)$ is the out-degree of node $v$) then it gives one
  grain along each of its out-going arcs. When the graph supporting the
  dynamics is not a lattice, \textit{sandpile models} are also called \textit{chip
  firing games} in the literature \cite{1991-Lovasz-CFG,1992-Lovasz-DirectedCFG}.
\end{myremark}

A cell $x$ is \textit{stable} if $c(x)<\theta$, and \textit{unstable} otherwise.
A configuration is \textit{stable} when all cells
are stable, and is \textit{unstable} if at least one cell is unstable. Remark that
stable configurations are fixed points of the global rule $F$. From the Equation~\eqref{eq:sandpile}, it is
clear that the system is deterministic and therefore given a configuration $c$ and a stable configuration
$c'$ either there is a unique sequence of configurations
$c = c_1 \donne c_2 \donne \dots \donne c_n = c'$ or 
$c\not\sdonne c'$. However, one can consider also other types updating policies. The \textit{sequential
policy} consists in choosing non-deterministically a cell from the unstable ones and
in updating only this chosen cell. Then, repeat the same update policy on the
newly obtained configuration and so on. It is clear that the new dynamics might be very different from the one
obtained from Equation~\eqref{eq:sandpile}. Sandpiles models in which the sequential update and the
parallel update policies produce the same set of stable configurations with the same number of topplings
are called \textit{Abelian}.
Recall that the terms
\textit{firing} and \textit{toppling} are employed to describe the action of moving
sand grains from unstable cells to other cells. The \textit{stabilization} of a configuration $c$
is the process of reaching a stable configuration. A \textit{finite}
configuration contains a finite number of grains, or equivalently its number of non-empty cells is finite.

The topplings counter, usually called \textit{shot vector} or \textit{odometer function}
in the literature, started at an initial configuration
configuration $c$ is a very useful formal tool for the analysis of sandpiles and it is defined as follows. 
For all configurations $c, c', c''$ if $c \donne c'$ then,
$$
  \forall x \in \Z^d : \odometer(c,c')(x) =
  \H{c(x)-\theta},
$$
and if $c \sdonne c' \donne c''$ then
$$
  \forall x \in \Z^d : \odometer(c,c'')(x) =
  \odometer(c,c')(x) + \odometer(c',c'')(x)
$$
with the convention $\odometer(c,c)(x)=0$ for all $x\in\Z^d$.

Let $c \donnes x c'$ denote application of the sequential update policy at cell $x$ \ie, 
one has $c \donnes x c'$ if and only if
\[
  c'(y)=\left\{
  \begin{array}{ll}
    c(y) - \theta\,\H{c(y)-\theta} &\text{ if }y=x\\
    c(y) + \distribution(y-x)\,\H{c(x)-\theta}
    & \text{ if } y \in x+\neighborhood\\
    c(y) & \text{ otherwise.}
  \end{array}
  \right.
\]
We also simply denote $c \donnes{} c'$ when there exists $x$ such that
$c \donnes x c'$. Moreover, let $\sdonnes$ denote
the reflexive transitive closure of $\donnes{}$, and
$\odometerseq(c,c')$ denote
the odometer function under the sequential update policy (it
counts the number of topplings occurring at each cell to reach
$c'$ from $c$).
The Abelian property can be formally stated as follows.

\begin{proposition}
  \label{prop:abelian}
  For any sandpile model, given a configuration $c$, if $c \donne c'$ then $c \sdonnes c'$.\\
  Moreover, if $c \sdonne c'$ and $c'$ is a stable configuration, then
  \begin{enumerate}
    \item $c \sdonnes c'$,
    \item $c \not\sdonnes c''$ for any other stable configuration $c''$,
    \item $\odometer(c,c')=\odometerseq(c,c')$.
  \end{enumerate}
\end{proposition}

We stress that Proposition~\ref{prop:abelian} is an important
feature in our context. Indeed, it states that the
(non-deterministic) sequential policy always leads to the same stable
configuration as the parallel policy,
with exactly the same number of topplings at each cell.
Relaxing this
requirement deeply changes the dynamics and the structure of the phase space (it has no more a lattice
structure for example). For non-abelian models see for example~\cite{formenti2006,formenti2007a}.
\smallskip

Endowing the set of configurations
with binary addition $+$ (given two configurations $c,c'$, $(c+c')(x)=c(x)+c'(x)$ for all $x$ \ie grain content is added cell-wise), 
$\N^{\Z^d}$ is a commutative monoid. It is also the case of the set of stable configurations, where the
  addition is defined as addition followed by stabilization.
The famous \textit{Abelian sandpile group} of
\textit{recurrent} configurations appears
when a \textit{global sink} is added to the multi-digraph supporting the
dynamics. This subject goes beyond the scope of the present survey,
for more see \cite{dhar1995algebraic,dhar1999abelian}.

\subsection{Prediction problems}

Given a sandpile model $\structure{\neighborhood,\distribution,\theta}$, the basic prediction problem
asks if a certain cell $x$ will become unstable when the system is started from a given finite initial configuration
$c$. More formally,

\vspace*{.5em}
\noindent\fbox{\parbox{\textwidth}{
  \textbf{Prediction problem ($\PRED$).}\\
  \textit{Input:} a finite configuration $c\in\cfgs$ and a cell $x\in\Z^d$.\\
  \textit{Question:} will cell $x$ eventually become unstable during the evolution from $c$?
}}
\vspace*{.2em}

One of the most intriguing features of sandpiles is that they are a paradigmatic example of
\textit{self-organized criticality}.
Indeed, starting from an initial configuration and adding
grains at random positions, the system reaches a stable
configuration $c$ from which a small perturbation (an
addition of a single grain at some cell) may trigger an
arbitrarily large chain of reactions commonly called an
{\em avalanche}. The distribution of sizes of avalanches
(when grains are added at random) follow a power
law~\cite{}. We are interested in the computational
complexity of deciding if a given cell $x$ will topple
during this process.
More
formally, one can ask the following.

\vspace*{.5em}
\noindent\fbox{\parbox{\textwidth}{
  \textbf{Stable prediction problem (\SPRED).}\\
  \textit{Input:} a stable finite configuration $c\in\cfgs$, two cells $x, y\in\Z^d$.\\
  \textit{Question:} $\!\!\!\!$\begin{tabular}[t]{l}does adding one grain on cell
  $y$ from $c$ trigger a chain of reactions\\that will eventually
  make $x$ become unstable?\end{tabular}
}}
\vspace*{.2em}

A variant of \SPRED is obtained when the cell $y$ is fixed from the very beginning.
When $y$ is the lexicographically minimal cell ({\em i.e.} the cell of the finite configuration with the lexicographically minimal coordinates), one has the
following.

\vspace*{.5em}
\noindent\fbox{\parbox{\textwidth}{
  \textbf{First column stable prediction problem (\FSPRED).}\\
  \textit{Input:} a stable finite configuration $c\in\cfgs$ and a cell $x\in\Z^d$.\\
  \textit{Question:} $\!\!\!\!$\begin{tabular}[t]{l}does adding one grain on the (lexicographically) minimal cell
  of $c$ trigger a \\ chain of reactions that will eventually make $x$ become unstable?\end{tabular}
}}
\vspace*{.2em}

Adding the grain at an extremity of the configuration
(the lexicographically {\em minimal} cell)
implies a strong monotonicity of the
dynamics, which has been
especially useful in one-dimensional proofs of $\NC$ness (see Subection
\ref{ss:avalanche} and Proposition~\ref{prop:fsmono})). $\FSPRED$ has also been called
the \textbf{Avalanche problem}.

Finally, one can simply ask how hard it is to compute the stable configuration
reached when starting from a given finite initial configuration. In other words,

\vspace*{.5em}
\noindent\fbox{\parbox{\textwidth}{
  \textbf{Computational prediction problem (\CPRED).}\\
  \textit{Input:} a finite configuration $c\in\cfgs$.\\
  \textit{Output:} what is the stable configuration reached when starting at $c$?
}}
\vspace*{.2em}

Before stepping to the detailed study of the complexity of the prediction
problems seen above, one shall discuss about the input coding and the
input size. Indeed, it is convenient
(at the cost of a polynomial increase in size) 
to consider that input configurations $c$ are given on finite $d$ dimensional
hypercubes of side $n$ placed at the origin, let us call them \textit{elementary
hypercubes} (they have a volume of $n^d$ cells). We also assume that the number of
sand grains stored at each cell of a configuration is strictly smaller than $2\theta$,
which is a constant, as this is an invariant (see Proposition~\ref{prop:theta}) that
\begin{itemize}
  \item allows to consider constant time basic operations,
  \item preserves the ``dynamical complexity'',
    within $\Poly$ and with $\Poly$-hard problems.
    Unbounded values bring considerations of another
    kind (namely, of computing fixed points from a
    single column of sand grains, as in
    \cite{levine2016,levine2017b,levine2017a,kevin-phd}) not
    necessary to capture the intrinsic complexity of
    the problem.
\end{itemize}

Cell positions are also admitted to be given using $\O(\log(n^d))$ bits which is
$\o(n^d)$ (see Lemma~\ref{lemma:configbound} for a polynomial bound on the most
distant cell that can receive a grain). The total size of any input is therefore
$\O(n^d)$, the total number of cells.

\begin{proposition}
  \label{prop:theta}
  For all $c\in\N^{\Z^d}$ and $x\in\Z^d$,
  if $c(x)<2\theta$ then $F(c)(x)<2\theta$.
\end{proposition}
\begin{proof}
  In one time step a cell $x$ gains at most $\theta$ grains (if all its
  in-neighbor topple), and it looses $\theta$ grains if $c(x)\geq\theta$ because
  it is unstable.
\end{proof}


\subsection{Avalanches}
\label{ss:avalanche}

For convenience sake, denote by $\indic_\set{y}$ the indicator function of cell $y\in\Z^d$,
so that adding one grain to cell $y$ of a configuration $c\in\cfgs$ translates into 
considering the configuration $c+\indic_\set{y}$.
\medskip

The notion of \textit{avalanche} naturally arises in the dynamics of sandpiles.
It represents the chain of reactions which originates from some grain addition
to a stable configuration.

\begin{definition}
  Given a stable configuration $c\in\cfgs$ and an index $y\in\Z^d$, the \textit{avalanche}
  generated by adding one grain at cell $y$ of $c$ is given by
  $\odometer(c+\indic_\set{y},c')$ where $c \sdonne c'$ and $c'$ is a stable configuration.
\end{definition}

Avalanches are especially related to \SPRED and \FSPRED where only one
grain is added to a stable configuration. In order to study the dynamics of
avalanches, it is useful to consider sequential iterations and to introduce a
canonical sequence of cell topplings (recall that under the sequential update policy, 
the system is non-deterministic).

\begin{definition}
  The \textit{avalanche process} associated with an avalanche
  $\odometer(c+\indic_\set{y},c')$ is the lexicographically minimal sequence
  $(z_1,\dots,z_t)$ such that:
  $$c+\indic_\set{y} \donnes{z_1} c^1 \donnes{z_2} \dots \donnes{z_t} c'.$$
  Remark that $t=\sum_{x \in \Z^d} \odometer(c+\indic_\set{y},c')(x)$.
\end{definition}

The avalanche corresponding to an instance of \FSPRED verifies the following strong monotonicity property, in the sense that the dynamics is very contrained and the odometer function is incremented at most once at every cell, until a stable configuration is reached.

\begin{proposition}
  \label{prop:fsmono}
  Given a finite configuration $c\in\cfgs$ within the elementary hypercube, the
  dynamics starting from $c+\indic_\set{0^d}$ to a stable configuration topples
  any cell at most
  once. Formally, for all $c'$ such that $c+\indic_\set{0^d} \sdonne c'$
  it holds
  \[
    \forall x \in \Z^d : \odometer(c+\indic_\set{0^d},c')(x)\in\set{0,1}.
  \]
\end{proposition}
\begin{proof}
  Let $z\in\Z^d$ be one of the chronologically first cells to topple twice, and $t_1$ and $t_2$
  be the times of the toppling event. Cell $z$ needs all its in-neighbors 
  to topple between
  times $t_1$ (included) and $t_2$ (excluded). If $z \neq 0^d$, then at
  least one in-neighbor $z'$ of cell $z$ is fired before it. This is a contradiction
  since $z'$ must topple for a second time before $t_2$ and $z$ was supposed 
  to be the first cell with that property. If $z = 0^d$, then
  we use the fact that $c$ is in an elementary hypercube and the grain addition
  is done at the origin: since the neighborhood spans the whole lattice
  (Equation \ref{eq:span}) there is an in-neighbor $z'$ of $z$ with $c(z')=0$,
  which cannot topple unless all its in-neighbors topple before it, and $z'$
  has an in-neighbor $z''$ with $c(z'')=0$ which in its
  turn cannot topple unless all its in-neighbors topple before it, \etc
  This leads to an infinite chain of consequences contradicting the fact that
  the dynamics converges to a stable configuration in a finite time (even in polynomial time, see Theorem~\ref{theorem:poly}).
\end{proof}

From the proof of Proposition~\ref{prop:fsmono}, one can also notice that,
for \SPRED, multiple topplings always originate
from cells starting in an unstable state (the sequential statement
is stronger).

\begin{proposition}
  \label{prop:smost}
  Given an instance $(c,x,y)$ of \SPRED, cell $y$ which receives a grain is
  always the most toppled cell throughout the evolution from
  $c+\indic_\set{y}$. Formally, for all $c'$ such that
  $c+\indic_\set{y} \sdonnes c'$ we have
  $\forall z \in \Z^d : \odometerseq(c+\indic_\set{y},c')(y)
  \geq \odometerseq(c+\indic_\set{y},c')(z)$.
\end{proposition}

\subsection{Complexity classes}
\label{ss:complexity}

This section quickly recalls the main definitions and results in complexity theory
that will be used in the sequel. For more details, the reader is referred 
to~\cite{limits,jaja92,sipser2012}.
\medskip

$\Poly$ is the class of decision problems solvable in polynomial time by a
deterministic Turing machine, or equivalently in polynomial time by a
random-access stored-program machine (RASP, a kind of RAM, \ie a
sequential machine with constant time memory access). For $i\in\N$, $\NC^i$ 
is the class  of decision problems solvable by a uniform family of Boolean circuits, with
polynomial size, depth $\O(\log^i(n))$, and fan-in 2, or equivalently in time $\O(\log^i(n))$ on a
parallel random-access machine (PRAM) using $\O(n^i)$ processors (it is not
important to consider how the PRAM handles simultaneous access to its shared
memory). For $i\in\N$, $\AC^i$ is the class of decision problems solvable by a non-uniform
family of Boolean circuits, with polynomial size, depth $\O(\log^i(n))$, and
unbounded fan-in. $\NC=\cup_{i \in \N} \NC^i$ and $\AC=\cup_{i \in \N} \AC^i$.
Some hardness results will also employ $\TC^0$, the class of decision problems
solvable
by polynomial-size, constant-depth circuits with unbounded fan-in, which can
use \textit{and}, \textit{or}, and \textit{not} gates (as in $AC^0$) as well as
threshold gates (a threshold gate returns 1 if at least half of its inputs are 1, and 0
otherwise).
To complete the picture, let us define the space complexity class $\L$ which consists in
decision problems solvable in logarithmic space on a deterministic Turing
machine, and $\NL$ its non-deterministic version. Classical relations among the above
classes can be resumed as follows (see~\cite{limits} for details):
\[
\NC^0 \subsetneq \AC^0 \subsetneq \TC^0 \subseteq \NC^1 \subseteq \L \subseteq \NL \subseteq \AC^1 \subseteq
\dots \NC^i \subseteq \AC^i \subseteq \NC^{i+1} \dots \subseteq \Poly.
\]
Intuitively, problems in $\NC$ are thought as \textit{efficiently computable in
parallel}, whereas $\Poly$-complete problems (under $\NC$ reductions or
below) are \textit{inherently sequential}. This distinction is interesting also
in the context of sandpiles: when can we efficiently parallelize the prediction?

In the PRAM model, processors can write the
output of the computation on their shared memory, and in circuit models, for
each input size there is a fixed number of nodes to encode the output of
the computation.  We denote $A \leq_C B$ when there is a many-one reduction
from $A$ to $B$ computable in $C$.
Remark that reductions in $\NC^0$ 
(constant depth and constant fan-in)
are the most restrictive we may consider:
each bit of output may depend only on a constant number of bits of input (for
example it cannot depend on the input size).
From a
decision problem point of view, the answer to a problem in $\NC^0$ depends only
on a constant part of its input. Computing the parity and majority of $n$ bits
is not in $\NC^0$, nor in $\AC^0$ as proved in~\cite{Furst1984}.

To give a lower bound on the complexity of solving a problem in parallel, a
$\TC^0$-hardness (under $\AC^0$ reduction) result means that the dynamics is
sufficiently complex to
perform non-trivial computation, such as the parity or majority
of $n$ bits (indeed, $\TC^0$ is the closure of {\bf Majority}\footnote{{\bf Majority} is the problem of deciding, given a word of $n$ bits, if it contains a majority of ones or not.} under constant depth
reductions). $\TC^0$-hard problems are not in $\AC^{1-\epsilon}$ for any
constant $\epsilon > 0$ \cite{Hastad87}.

In this last part of the section we are going to recall
some notable open questions in complexity theory related to the classes seen so far
and try to connect them with our prediction problems.

\begin{open}
  $\NC \neq \Poly$? (It is not even known whether $\NC \neq \NP$ or
  $\NC = \NP$.)
\end{open}

\begin{open}
  Are $\NC^i$ and $\AC^i$ proper hierarchies of classes or do they collapse at some level $i\in\N$?
\end{open}

The \textbf{Circuit value
problem} (\CVP) is the canonical $\Poly$-complete problem (under $\AC^0$ reductions): predict the output of the
computation of a given a circuit with identified input gates and one output gate. It remains $\Poly$-complete
when restricted to \textbf{monotone} gates (\MCVP), when restricted
to \textbf{planar} circuits ($\PCVP$), but not both: \MPCVP is in
$\LOGCFL \subseteq \AC^1$ \cite{cook89,yang1991}.

Since the early studies of Banks in \cite{banks-phd}, the reduction from \MCVP is
the most widespread method (if not the only one) to prove the $\Poly$-completeness
of prediction problems in discrete dynamical systems (in particular for sandpiles). 
This reduction technique is often referred to as \textit{Banks' approach} (see
Section~\ref{s:pcomplete} 
for applications of Banks' approach to sandpiles).

There are obvious reductions among the decision versions of the prediction
problems, giving a hierarchy of difficulties.

\begin{proposition}
  \label{prop:hierarchy}
  $\FSPRED \red{\AC^0} \SPRED \red{\AC^0} \PRED$.
\end{proposition}

The functional version of the prediction problem, \CPRED, seems harder
than just answering a yes/no question about one cell. It 
relates a function problem to a decision problem. We propose a clear
statement.

\begin{conjecture}
  \label{conj:func}
  $\PRED \in {\NC^1}^\CPRED$. In other words, the decision problem \PRED can be
  solved in $\NC^1$ with a \CPRED oracle (which is a function problem).
\end{conjecture}

In the opposite direction (relating the complexity of a harder problem to an easier
one), Proposition~\ref{prop:smost} hints at a decomposition of the prediction of
an arbitrary sandpile into a succession of avalanches. However, it is not clear
if the hierarchy of problems is proper or if some decision problem are
equivalent in terms of computational difficulty. 
To our knowledge there is no example of a sandpile model for
which the computational complexity of any two of these problems
would be different.

\begin{open}
  Is the hierarchy given in Proposition~\ref{prop:hierarchy} and Conjecture~\ref{conj:func}
  proper? In other terms, are there sandpile models such that \CPRED (resp.
  \PRED, \SPRED) is strictly harder than \PRED (resp. \SPRED, \FSPRED)?
\end{open}

\section{All prediction problems are in $\Poly$}
\label{s:poly}

The starting point of complexity studies in sandpiles is a paper of Moore and
Nilsson in $1999$. It gives the global picture~\cite{1999-Moore-complexitySandpiles}
(for von Neumann sandpile model and $\CPRED$): 
sandpile prediction is in $\NC$ in dimension one, and it is $\Poly$-complete from 
dimension three and above. The two-dimensional case is somewhat surprisingly open. 
Many results appeared since ~\cite{1999-Moore-complexitySandpiles} made the picture
more precise. In this section we prove that prediction problems for all sandpile models are in
$\Poly$ regardless of the dimension, because the sandpile dynamics runs for a
polynomial number of steps before it stabilizes.
\smallskip

One easily gets the intuition 
that from any finite configuration, the dynamics converges to a stable
configuration because sand grains spread all over the lattice (Equation
\ref{eq:span}). In~\cite{1988-Tardos-PolyBoundCFG}, Tardos proved a polynomial
bound on the convergence time when the graph supporting the dynamics is finite
and undirected (the original bound is $2vek$ for a graph consisting of $v$ vertices, $e$
edges, and having diameter $k$, when the dynamics converges to a stable
configuration). The proof idea generalizes, not so trivially, using a series
of lemma as follows.

\begin{lemma}
  \label{lemma:bij}
  For any finite non-empty set $X \subset\Z^d$ of cells, there exists a bijection
  $\lambda\colon X_\text{in}\to X_\text{out}$ with
  \begin{eqnarray*}
    X_\text{in}&=&\set{ (y,x) \mid y \notin X, x \in X \text{ and } y \in
      x-\neighborhood}\text{ (arcs from $^c\!X$ to $X$),}\\
    X_\text{out}&=&\set{ (x,y) \mid x \in X, y \notin X \text{ and } y \in
      x+\neighborhood}\text{ (arcs from $X$ to $^c\!X$),}
  \end{eqnarray*}
  such that $\lambda((y,x))=(x',y')$ implies $x+x'=y+y'$. 
\end{lemma}
\begin{proof}
  For all $(y,x) \in X_\text{in}$, the bijection $\lambda$ is defined as
  $$\lambda((y,x))=(x+k^*(x-y),x+(k^*+1)(x-y))$$
  with $k^*=\min\set{ k\in\N\mid (x+k(x-y),x+(k+1)(x-y)) \in X \times \,^c\!X}$. 
  First of all, remark that $\lambda$ is well defined. Indeed,
  for any edge $(y,x)$ in $X_\text{in}$, since we are considering a lattice and since $X$ is finite, 
  there must exist  a path containing $(y,x)$ and passing through some edge $(x',y')$ belonging
  to $X_\text{out}$. Then, $k$ is the number of edges between $(y,x)$ and $(x',y')$.
  Finally, $\lambda$ is 
  bijective since
  \[
      \mathcal R=\set{((y,x),(x+k(x-y),x+(k+1)(x-y))) \mid k\in\Z}
  \]
  is an equivalence
  relation, and the equivalence classes verify
  $|[(y,x)]_{\mathcal R} \cap X_\text{in}|$ equals
  $|[(y,x)]_{\mathcal R} \cap X_\text{out}|$ and is finite,
  for any $y,x$ (this is the number of times
  the associated line enters
  and exits $X$).
\end{proof}

The following lemma is stronger when expressed in the sequential context.

\begin{lemma}
  \label{lemma:topplingsbound}
  For any pair of finite configurations $c, c'\in\cfgs$ such that $c\sdonnes c'$, and pair of cells
  $x,y\in\Z^d$ such that $y \in x+\neighborhood$, it holds
  \[
     |\odometerseq(c,c')(x) - \odometerseq(c,c')(y)| \leq \sum_{z \in \Z^d}c(z).
  \]
\end{lemma}

\begin{proof}
  Suppose $\odometerseq(c,c')(x) < \odometerseq(c,c')(y)$ (the other case is
  symmetric). Let
  $$X=\{ z\in\Z^d \mid \odometerseq(c,c')(z) \leq \odometerseq(c,c')(x)\}$$
  be the set of cells that toppled at most as much as $x$ and define $X_\text{in}$,
  $X_\text{out}$ and the bijection $\lambda$
  as in the proof of Lemma~\ref{lemma:bij}. We have $x \in
  X$ and $y \notin X$ therefore $(x,y) \in X_\text{out}$. Then, using
  $\lambda$, one can count the number of grains that moved in and out of $X$,
  \begin{align*}
    \sum_{z \in X} c'(z) &\geq
    \sum_{(y',x') \in X_\text{in}} \odometerseq(c,c')(y')\,\distribution(x'-y')
    - \sum_{(x',y') \in X_\text{out}} \odometerseq(c,c')(x')\,\distribution(y'-x')\\
    &= \sum_{\overset{\scriptstyle (y',x') \in X_\text{in}}
    {\lambda((y',x'))=(x'',y'')}}
    \Big[\odometerseq(c,c')(y')\,\distribution(x'-y')
    -\odometerseq(c,c')(x'')\,\distribution(y''-x'')\Big].
  \end{align*}
  Since $\lambda((y',x'))=(x'',y'')$ implies $x'+x''=y'+y''$, for each term of
  the last sum we have $\distribution(x'-y')=\distribution(y''-x'')$, and by
  definition of $X$, one finds $0 \leq \odometerseq(c,c')(y') \geq
  \odometerseq(c,c')(x'')$, therefore each term is positive. The result follows
  because in $c'$ there cannot be more grains at cells in $X$
  than total number of sand grains in $c$ ($\sdonnes$ is number-conserving).
\end{proof}

The last argument in the proof of Tardos \cite{1988-Tardos-PolyBoundCFG} is
concerned with the finiteness of
the underlying graph and, of course, it does not apply here. However, it is possible to
bound the region of the lattice that may eventually receive at least one sand
grain.

\begin{lemma}
  \label{lemma:configbound}
  For any finite configuration $c\in\cfgs$ belonging to the elementary
  hypercube of size $n^d$, if $c \sdonne
  c'$
  for some $c'\in\cfgs$
  , then $c'(x)=0$ for any $|x|_\infty> 4 \theta r n^d$, with
  $r=\max\set{ |v|_\infty \mid v \in \neighborhood}$.
\end{lemma}

\begin{proof}
  The configuration $c$ may contain at most $(2\theta-1)n^d$ sand grains since it belongs to the elementary
  hypercube of size $n^d$. Let us prove that:
  \begin{enumerate}
    \item sand grains cannot all leave one place,
    \item there is no isolated sand grain.
  \end{enumerate}
  Consider an encompassing hypercube $\mathcal R=\set{-r,\dots,r}^d$ 
  around the neighborhood $\neighborhood$. Let us first show that, 
    $$\text{if for some } x\in\Z^d, c(x)>0 \text{ then } \exists y \in x+\mathcal R \text{ such that } c'(y)>0.$$ 
  By contradiction,
  consider any of the cells that where the last to topple inside $x+\mathcal
  R$, from Equation \ref{eq:span} it must have send at least one grain inside
  $x+\mathcal R$ which thus cannot be empty. An analogous argument shows that,
  $$
    \text{if } c'(x)>0 \text{ then } \exists y \in x+2\mathcal R
    \text{ such that } c'(y)>0.
  $$
  We can therefore conclude that in $c'$ there is a grain
  within cells $\{-r,\dots,n+r\}^d$, and any other grain cannot be at distance
  greater than $(2\theta-1)n^d2r$ in the max norm.
\end{proof}

Now the previous lemmas can be exploited to fit the argumentation of 
Tardos~\cite{1988-Tardos-PolyBoundCFG} in this general framework.

\begin{theorem}
  \label{theorem:poly}
  Given any finite configuration $c\in\cfgs$ of size $n^d$, a stable configuration is
  reached within at most $(2 d \theta r n)^{\O(d^2)}$ time steps,
  a polynomial in the size of $c$.
\end{theorem}
\begin{proof}
  According to Lemma~\ref{lemma:configbound}, two cells which have toppled cannot be at
  distance (through a chain of neighbors in the $\ell_1$-norm) greater than $8
  d \theta r n^d$. By repeated applications of Lemma~\ref{lemma:topplingsbound},
  one finds that no cell can topple more than $8 d \theta r n^d \sum_{z\in\Z^d}
  c(z)$ times. Since no more than $(8 d \theta r n^d)^d$ different cells can
  topple (Lemma~\ref{lemma:configbound} again), and the number of grains is
  upper bounded by $2\theta\,n^d$, we get precisely
  $2^{3d+4}d^{d+1}\theta^{d+2}r^{d+1}n^{d^2+2d}$ which is upper-bounded for any $d$ by
  $(2 d \theta r n)^{\O(d^2)}$.
\end{proof}

Theorem~\ref{theorem:poly} provides an upper bound on all the prediction problems
for any sandpile model.

\begin{corollary}
  $\PRED, \SPRED, \FSPRED, \CPRED \in \Poly$.
\end{corollary}

\section{$\NC$ in dimension one}
\label{s:1d}

In one dimension, prediction problems on sandpile dynamics have been proven to
be efficiently computable in parallel, \ie they lie in $\NC$. As mentioned in
the introduction of Section~\ref{s:poly}, the whole story began in 1999 with a study of the
computation variant of the prediction problems, and it has been successively extended to more
restrictive variants.

\begin{theorem}[\cite{1999-Moore-complexitySandpiles}, improved in \cite{Miltersen2005}]
  For von Neumann sandpile model of radius one in dimension one (see Remark~\ref{remark:vn}),
  \CPRED is in $\NC^1$,
  and it is not in $\AC^{1-\epsilon}$ for any constant $\epsilon>0$.
\end{theorem}

The last part comes from a simple constant depth reduction of the \textbf{Majority} of $n$ bits
problem (given $x \in \B^n$, decide if there is a majority of one)


to sandpile dynamics, proving that the problem is $\TC^0$-hard.
These results rely on a clever technical study of predicting the dynamics of a
stripe of 1s containing a single $2$ (because $\theta=2$ in this model): a 0 appears
within the stripe of 1s which is enlarged,
such that the center of mass is unchanged. Generalizing
these results seems to be a technically challenging task, but we conjecture
that they do.

\begin{conjecture}
  \label{conj:1d}
  For any one-dimensional sandpile model, prediction is efficiently computable
  in parallel, \ie $\CPRED, \PRED, \SPRED, \FSPRED\in\NC$.
\end{conjecture}

Three results from the literature support this conjecture, they are
expressed on $\FSPRED$
and exploit the strong monotonicity of avalanches in this case (Proposition~\ref{prop:fsmono})
to prove that the problem is in $\NC^1$ for a large
class of models:
\begin{itemize}
  \item Kadanoff sandpile models (\cite{2010-FormentiGolesMartin-KSPMAP}
    completed in \cite{fpr14}),
  \item extended to any decreasing sandpile model (in \cite{fpr18}).
\end{itemize}

\begin{myremark}
  \label{remark:kadanoff}
  In dimension one, the radius $r$ {\em Kadanoff} sandpile model corresponds to
  $\structure{\neighborhood_K,\distribution_K,r+1}$ with $\neighborhood_K=\{-1,r\}$
  and $\distribution_K(-1)=r, \distribution_K(r)=1$ (see \cite{fpr14} for an
  illustration). A {\em decreasing} sandpile
  model simply has a neighborhood $\neighborhood$ such that $\neighborhood \cap
  \N_- = \{-1\}$.
  The name decreasing sandpile model comes from an
  interpretation of the sand content at each cell as the slope between
  consecutive columns of sand grains (it adds one artificial dimension to the
  picture), such that it preserves a monotonous form.
\end{myremark}

\begin{theorem}[\cite{2010-FormentiGolesMartin-KSPMAP,fpr14,fpr18}]
  For any one-dimensional decreasing sandpile model,\linebreak
  $\FSPRED \in \NC^1$.
\end{theorem}

This result can be generalized to
any one-dimensional sandpile model (Theorem~\ref{theorem:1d}), which supports Conjecture \ref{conj:1d}.
However the generalization of the idea presented in \cite{fpr14,fpr18}
is not straightforward. Decreasing sandpile models
have the feature that $\neighborhood \cap \N_- = \{-1\}$ which, together with
the strong monotonicity Proposition~\ref{prop:fsmono}, implies a pseudo-linear
dynamics of the avalanche process (lexicographically minimal sequence of
topplings under the sequential update policy): using a sliding window of width
$r=\max\{|v|_\infty \mid v \in \neighborhood\}$, one can compute the whole avalanche 
from cell $0$ to the maximal index of a toppled cell. Then,
it is possible to precompute the topplings locally (via functions of constant
size telling what happens around cell $i+1$, called \emph{status} at $i+1$ in
\cite{fpr14,fpr18}, according to what happens around cell $i$, \ie from the \emph{status}
at $i$), and compose these informations according to a binary tree of logarithmic height, 
each level of the composition being computed in constant time, hence resulting in an
$\NC^1$ algorithm.

This pseudo-linear dynamics does not hold any more for
general sandpile models in dimension one, as shown in the example in
Figure~\ref{fig:1d-nonlinear}. Nevertheless, it is possible to get a similar
result with a more involved construction presented below.

\begin{figure}
  \centering{\includegraphics{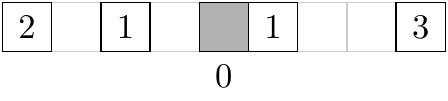}
  \hspace*{.5cm}
  \includegraphics{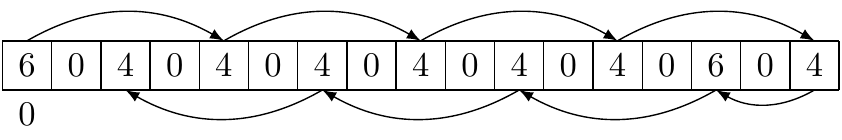}}
  \caption{Example of avalanche process for a one-dimensional sandpile model
  (left, $\theta=7$), which is not pseudo-linear: the avalanche process may
  topple cells arbitrarily far on the right before going backwards to topple
  cells on the left end (right, instance of $\FSPRED$ where arrows depict the
  avalanche process).}
  \label{fig:1d-nonlinear}
\end{figure}

In dimension one, for any $c, c' \in \N^{\Z}$,
let $\odometer(c,c')|_{[x,y]}$ denote the restriction of $\odometer(c,c')$ to the
interval $[x,y]$, with $x<y$ two cells in $\Z$.

\begin{lemma}
  \label{lemma:odor}
  For any one-dimensional sandpile model
  $\structure{\distribution,\neighborhood,\theta}$ and any configurations
  $c,c'$ such that $c$ is within the elementary hypercube, $c \sdonne c'$ and
  $c'$ is stable,
  with $r=\max\{|v|_\infty \mid v \in \neighborhood\}$,
  knowing
  \begin{itemize}
    \item $\odometer(c,c')|_{[x-r,x-1]}$ and
      $\odometer(c,c')|_{[y+1,y+r]}$
  \end{itemize}
   for cells $x,y \in \Z$ such that $x+r \leq y$, allows to compute
  \begin{itemize}
    \item $\odometer(c,c')|_{[x,x+2r-1]}$ and $\odometer(c,c')|_{[y-2r+1,y]}$
  \end{itemize}
  in time $\O(y-x)$ on one
  processor.
\end{lemma}

Figure~\ref{fig:odor} provides a graphical illustration of the
statement of Lemma~\ref{lemma:odor}.

\begin{figure}
  \centering{\includegraphics{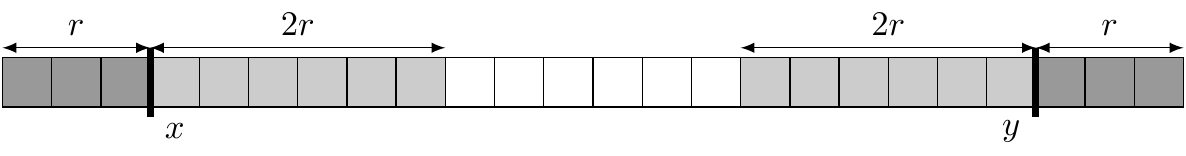}}
  \caption{When every cell topples at most once, knowing which cells among
  $[x-r,x-1]$ and $[y+1,y+r]$ topple allows to compute which cells among
  $[x,x+2r-1]$
  and $[y-2r+1,y]$ topple, where $r=\max\set{|v|_\infty \mid v \in \neighborhood}$
  (Lemma~\ref{lemma:odor}).}
  \label{fig:odor}
\end{figure}

\begin{proof}
  If one knows all topplings that may influence the topplings within $[x,y]$ (given by the
  assumptions $\odometer(c,c')|_{[x-r,x-1]}$ and $\odometer(c,c')|_{[y+1,y+r]}$,
  where $r$ is the radius of the sandpile model), it is possible
  to compute all topplings occurring within $[x,y]$. Furthermore, from
  Proposition~\ref{prop:fsmono}, any cell topples at most once and hence the number of
  topplings is upper bounded by $y-x$, which in its turn is an upper bound on the number of
  computation steps. The condition $x+r<y$ ensures that the output is made of
  $\odometer(c,c')$ values within the interval $[x-r,y+r]$.
\end{proof}

\begin{theorem}
  \label{theorem:1d}
  For any one-dimensional sandpile model, $\FSPRED \in \NC^1$.
\end{theorem}

\begin{proof}[Sketch]
  Let us describe how to derive an $\NC^1$ algorithm for $\FSPRED$ from
  Lemma~\ref{lemma:odor}. Let $c,x$ be the instance, and $c'$ be the
  stable configuration such that $c \sdonne c'$. For simplicity sake, assume
  that $c$ is $c+\indic_{\set{0}}$,
  and let $r=\max\set{|v|_\infty \mid v \in \neighborhood}$.
  \begin{enumerate}
    \item Compute, in parallel for every $y \in \N$ multiple of $r$, the
      function taking as input $\odometer(c,c')|_{[y-r,y-1]}$ and
      $\odometer(c,c')|_{[y+4r,y+5r-1]}$, and outputting $\odometer(c,c')|_{[y,y+2r-1]}$
      and $\odometer(c,c')|_{[y+2r,y+4r-1]}$ (see picture below
      for a partial example with some $y$ multiple of $r$ and $z=y+4r$, on two copies of the configuration for clarity).
      From Lemma~\ref{lemma:odor}, each computation needs a constant time on
      one processor, and there are linearly many.\\
  \end{enumerate}
  \begin{center}
    \includegraphics[width=\textwidth]{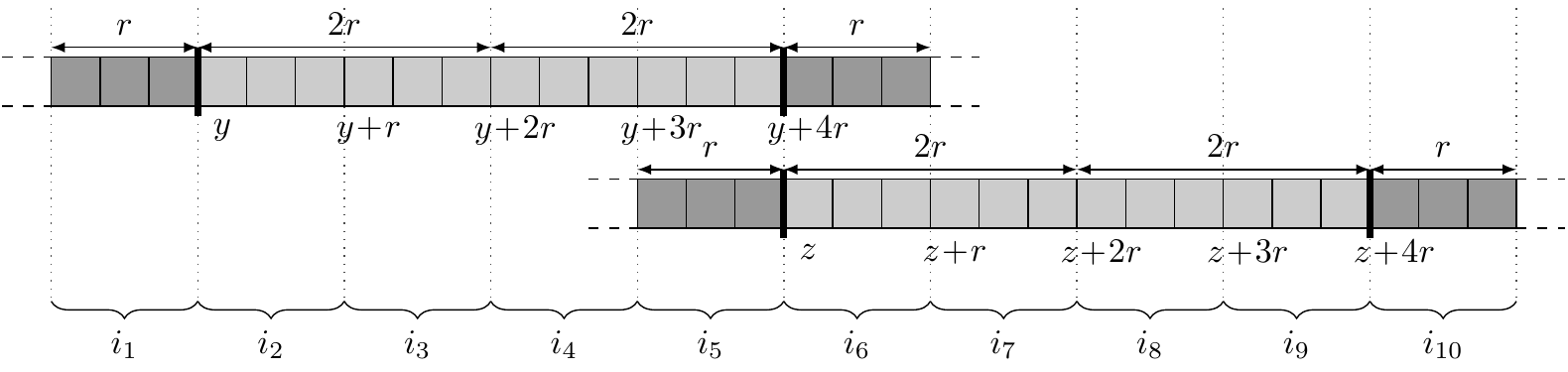}
  \end{center}
  \begin{enumerate}[resume]
    \item Compose them according to some binary tree: each function is of constant
      size (from Proposition~\ref{prop:fsmono}, input is $2r$ bits and output is
      $4r$ bits), and we compose functions having some carefully chosen
      overlap. Let us denote $i_1,i_2,i_3,i_4,i_5,i_6,i_7,i_8,i_9,i_{10}$ the
      respective portions of $\odometer(c,c')$ of size $r$ under consideration
      (see picture above), then the two functions are respectively: $i_1,i_6
      \mapsto i_2,i_3,i_4,i_5$ and $i_5,i_{10} \mapsto i_6,i_7,i_8,i_9$. With
      these two we can compute $i_1,i_5,i_6,i_{10} \mapsto
      i_2,i_3,i_4,i_5,i_6,i_7,i_8,i_9$, and fortunately the fixed point
      regarding $i_5,i_6$ is uniquely determined by $i_4,i_7$,
      according to a constant time application of
      Lemma~\ref{lemma:odor} (remark that there may be an
      arbitrarily large gap between $i_3$ and $i_4$, and
      also between $i_7$ and $i_8$). Each
      composition deals with a constant number of functions of constant size,
      hence it takes a constant time. As there is a logarithmic number of
      levels in the composition, the overall process takes a logarithmic
      parallel time.
  \end{enumerate}
  Now remark that no cell within $[-r,-1]$ nor $[n+1,n+r]$ topple, with $n$ the
  size of configuration $c$, hence $\odometer(c,c')|_{[-r,-1]}$ and
  $\odometer(c,c')|_{[n+1,n+r]}$ are fixed.
  In order to get the answer of whether cell $x=kr+k'$ topples
  (for some unique $k,k' \in \N$ with $0 \leq k' < r$), we consider two
  binary trees of compositions: one such that the root gives the function with
  input $[-r,-1],[kr,(k+1)r-1]$, and the other such that the root gives the
  function with input $[(k-1)r,kr-1],[n+1,n+r]$. The fixed point of their
  conjunction, resolved with the function with input
  $[(k-2)r,(k-1)r-1],[(k+1)r,(k+2)r-1]$
  (again a constant time application of Lemma~\ref{lemma:odor}),
  tells whether
  $\odometer(c,c')(x)$ equals $0$ or $1$ (see picture below).
  \begin{center}
    \includegraphics[width=\textwidth]{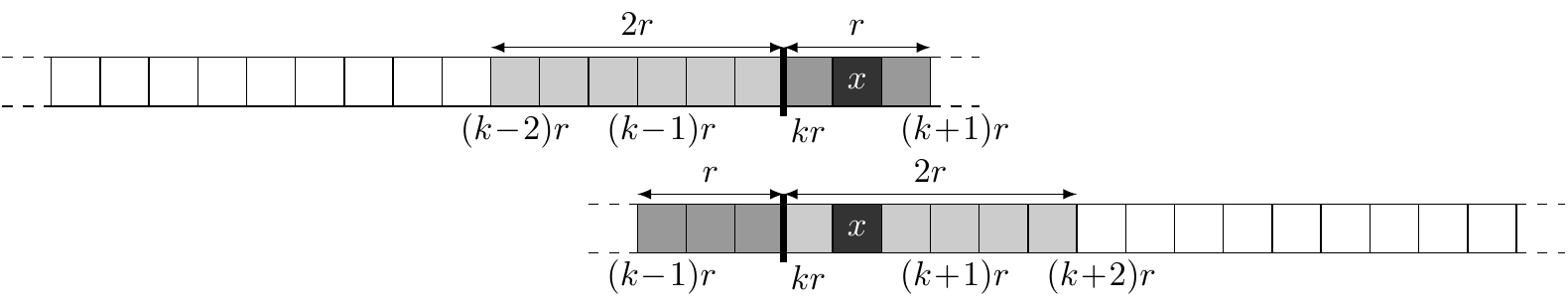}
  \end{center}
\end{proof}

As mentioned above, the generalization of the Theorem~\ref{theorem:1d} to other prediction
problems requires non-trivial extensions, because of multiple topplings at a cell which somehow
should be handled in constant time computation. The proof of $\TC^0$-hardness
from \cite{Miltersen2005} also makes heavy use of multiple topplings, and as a consequence
it does not generalize for free to an arbitrary one-dimensional
sandpile model. We nevertheless conjecture that it also does.

\begin{conjecture}
  For any one-dimensional sandpile model,
  $\PRED$ is $\TC^0$-hard for $\AC^0$
  reductions, and therefore not in
  $\AC^{1-\epsilon}$
  for any constant $\epsilon>0$.
\end{conjecture}

\section{$\Poly$-completeness in dimension three and above}
\label{s:pcomplete}\label{s:3d}

During his PhD thesis in the 1970s, Banks
started to implement circuit computation using 
discrete dynamical systems working on grids (namely cellular automata)
\cite{banks-phd}. The intuition behind such implementations 
is quite straightforward: a sequence of cells change state in a 
chain of reactions to transport information; two flows of information can interact to create logic gates.

For simplicity, circuits are restricted to have the following 
characteristics:
\begin{itemize}
  \item gates have fan in and fan out $2$,
  \item layered (information flows from one end to the other, layer by layer from inputs to output,
    without going backward),
  \item arranged on a grid layout.
\end{itemize}
Even with the previous constraints the prediction problems 
($\textsf{SAM2CVP}$ in \cite{limits}) are still $\Poly$-complete.
Moore and Nilsson ported Banks' technique to sandpile models in 1999
\cite{1999-Moore-complexitySandpiles}. They used an adaptation to the 
grid of the computation encoding idea developed by Bitar, Goles and Margenstern in
\cite{bitar1992,1996-GolesMargenstern-SPMUniversal}.

\begin{theorem}[\cite{1999-Moore-complexitySandpiles}]
  \label{theorem:VN3d}
  For von Neumann neighborhood $\neighborhoodVN$ of radius one and dimension $d
  \geq 3$, the problem $\FSPRED$ is $\Poly$-complete.
\end{theorem}

\begin{proof}[Sketch]
  The reduction is from $\MCVP$ (see Section \ref{ss:complexity}), in the case of fan in and fan out 2, layered,
  on a grid. We will present the proof for dimension 3 in details, and explain
  at the end how it generalizes.

  There are different types of gadgets to implement within
  sandpiles: wires, turns, \textit{and} gates, \textit{or} gates, plus diodes (to
  prevent unintended backward propagation of information) and multipliers (to
  get fan out 2). All these elements are presented as macrocells
  in Figure~\ref{fig:VNgates}
  (two top rows),
  and are embedded in a two-dimensional plane of the three-dimensional
  configuration. The non-planarity of the circuit requires (recall that $\MPCVP$ is
  in $\NC$, hence it is important that the circuit may not be planar) that
  wires cross each other, which is achieved using the
  third dimension (Figure~\ref{fig:VNgates} bottom row).
  
  \begin{figure}
    \centerline{\begin{tabular}{ccc}
      \includegraphics[scale=.8]{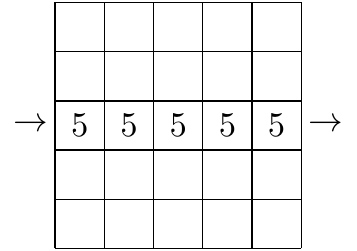} &
      \includegraphics[scale=.8]{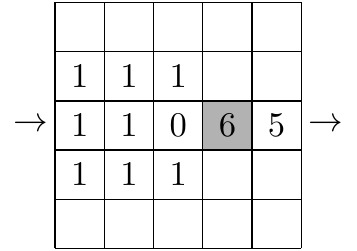} &
      \includegraphics[scale=.8]{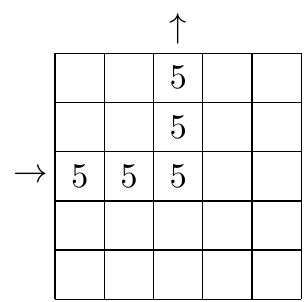}\\
      wire (0 or idle) &
      wire (transmits 1) &
      turn
    \end{tabular}}
    \centerline{\begin{tabular}{cccc}
      \includegraphics[scale=.8]{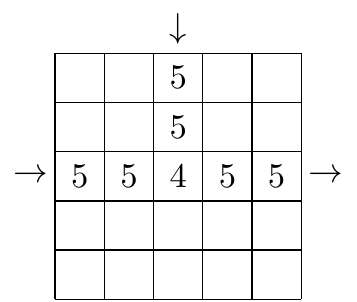} &
      \includegraphics[scale=.8]{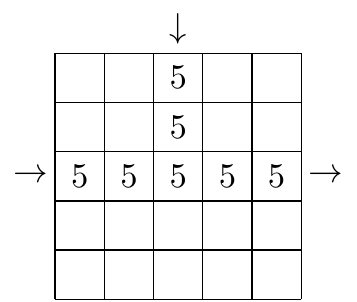} &
      \includegraphics[scale=.8]{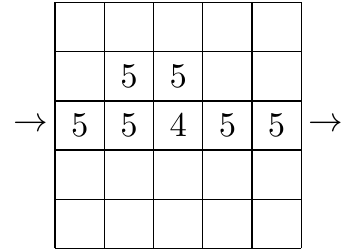} &
      \includegraphics[scale=.8]{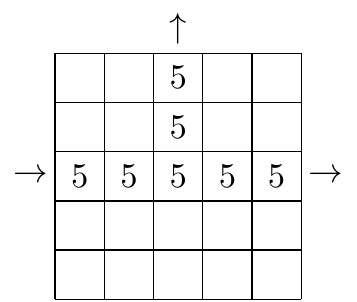}\\
      {\em and} gate &
      {\em or} gate &
      diode &
      multiplier
    \end{tabular}}
    \centerline{\begin{tabular}{ccc}
      \raisebox{1em}{\includegraphics[scale=.8]{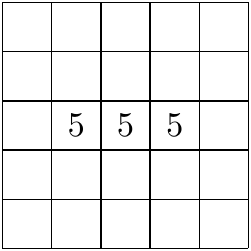}} &
      \includegraphics[scale=.8]{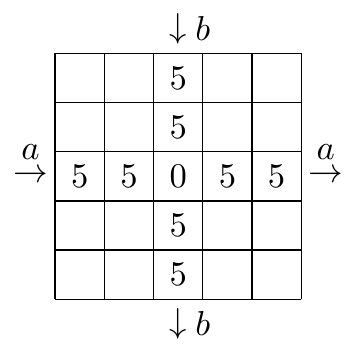} &
      \raisebox{1em}{\includegraphics[scale=.8]{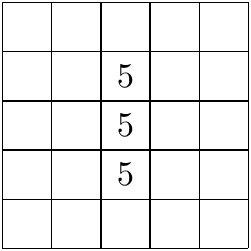}}\\
      cross-over & cross-over & cross-over\\
      ($-1$ in $3^\text{rd}$ dimension) &
      ($0$ in $3^\text{rd}$ dimension) &
      ($1$ in $3^\text{rd}$ dimension)
    \end{tabular}}
    \caption{Implementation of wires and gates with von Neumann neighborhood of
    radius one, in three dimensions ($\theta=6$).
    Only the cross-over gate uses the third dimension
    (it should be thought as three stacked layers),
    all other gates can be considered to live at coordinate $0$ in the third
    dimension.
    A wire is implemented as a
    simple chain of reactions of cell topplings: when it is triggered it
    transmits a 1, if not (\ie it remains stable) it transmits a 0.
    Arrows indicate fan in and fan out (input and output pairs are identified
    in the cross-over gate: a signal coming from input $a$ (resp. $b$) results
    in a signal going to output $a$ (resp. $b$),
    independently of each other), empty cells
    contain no sand grain.
    Remark that appart from central cells of cross-overs,
    initialy empty cells will not topple.}
    \label{fig:VNgates}
  \end{figure}

  Now the idea is that by replacing circuit elements with macrocells we
  can have a single grain addition triggering a
  wire that can be multiplied to implement
  the input constants on the first layer, which can be connected to gate on the second
  layer, {\em etc}, until the output gate which is simply a wire with the
  questionned cell in the center. As noted in \cite{Moore1997}, diodes should
  be added between layers to prevent backward propagation of information
  (specificaly to prevent false 1).

  To finish the reduction, we would like to insist on the fact that crossing of
  wires in the third dimension is only necessary to overcome the non-planarity
  of $\MCVP$. In fact, it is equivalent to have a crossing or a negation gate:
  \begin{itemize}
    \item In \cite{limits} on $\PCVP$: ``\emph{A planar \emph{xor} circuit can be
      built from two each of \emph{and}, \emph{or}, and \emph{not} gates; a planar
      \emph{cross-over} circuit can be built from three planar \emph{xor}
      circuits}''.
    \item In \cite{1999-Moore-complexitySandpiles} Moore and Nilsson wrote:
      ``{\em Using a {\em double-wire} logic where each variable corresponds to
      a pair of wires carrying $x$ and $\bar{x}$, [Bitar, Goles and
      Margenstern] implement negation as well by crossing the two
      wires}''\footnote{interestingly, the next sentence in this quote is:
      ``{\em Since this violates planarity, negation does not appear to be
      possible in two dimensions}''.}.
  \end{itemize}
  We prefer to imagine a {\em single-wire} logic, with {\em cross-over} of
  wires in the third
  dimension. Cell topplings correspond to the circuit computation (truth value
  1 transmitted from layer to layer, going through gates), and the circuit
  outputs 1 if and only if the questioned cell topples.
  
  This reduction is performed in constant parallel time (in $\AC^0$):
  each constant size part of the circuit
  is converted to a sandpile sub-configuration (macrocell) of constant size placed at a
  fixed position, hence in a PRAM model each processor can handle one such part
  (there are polynomially many) in constant time.

  In order to generalize the proof to any dimension $d \geq 3$, simply remark
  that the same construction can be embedded in only three dimensions
  among many, with
  coordinate zero in all other dimensions, provided one adapts the sand content of
  non-empty cells according to $\theta$ (plus two grains for any dimension above
  three).
\end{proof} 

From the proof above we can notice that to simulate a circuit (instance of
$\MCVP$), von Neumann sandpile model of radius one and dimension $d \geq 3$ undergoes a
somewhat simple dynamics:
\begin{itemize}
  \item only three dimensions are used (mainly two),
  \item only cell contents $0$, $\theta-2$ and $\theta-1$ appear in macrocells,
  \item the strong monotonicity of Proposition~\ref{prop:fsmono} holds (any cell
    topples at most once),
  \item we have wires and gates organized in successive layers from inputs to
    output,
  \item all directions of information are known in the reduction,
  \item there are diodes everywhere so that information flows in exactly
    one direction.
\end{itemize}
As a consequence we can consider that the evolution of the obtained $\FSPRED$
instance $(c,x)$ verifies very restrictive conditions.
Let us consider that there is a diode between every pair of gates presented on
Figure~\ref{fig:VNgates}. Then the flow of information between gates and layers
is totally fixed, and we can almost claim that for any pair of neighbouring
cells it is known which one would topple first (in the case both topple).
However this is not completely accurate, as in {\em or} gates for example: if only
one of the two inputs transmits a 1 signal, then some cells going to the other
input topple ``backward''. This is clearly not an issue nor an important
feature.
We formalize how another sandpile model with $d \geq 3$ can perform the same
kind of circuit simulation as von Neumann of radius one in three dimensions.

\begin{lemma}
  \label{lemma:3d}
  If a sandpile model
  of neighborhood $\neighborhood$ has three linearly independent $x,y,z \in
  \neighborhood$ such that $ax+by+cz \notin \neighborhood$ for any $a,b,c \in
  \Z$, except when:
  $$a=\pm1, b=0, c=0 \quad\text{or}\quad a=0, b=\pm1, c=0 \quad\text{or}\quad
  a=0, b=0, c=\pm1,$$
  then it has a $\Poly$-complete $\FSPRED$ problem.
\end{lemma}

\begin{proof}
  Let us add two modifications to the circuit simulation of von Neumann radius
  one in three dimensions presented in the proof of Theorem~\ref{theorem:VN3d}, so that it
  fits the present context.
  First, up to a straightforward layout shift, we can embed all planar gates
  from Figure~\ref{fig:VNgates} with only two directions of information
  transmission: $x$ and $y$.
  Second, the third direction $z$ is used for the
  cross-over, which shifts in this direction the rest of the sandpile
  implementation of the circuit (see Figure~\ref{fig:3d}). We can simply assume
  that there is one cross-over per layer to fix the $z$-shift everywhere.
  These
  coordinate changes do not change the computational complexity of the problem.

  \begin{figure}
    \centerline{\includegraphics[width=\textwidth]{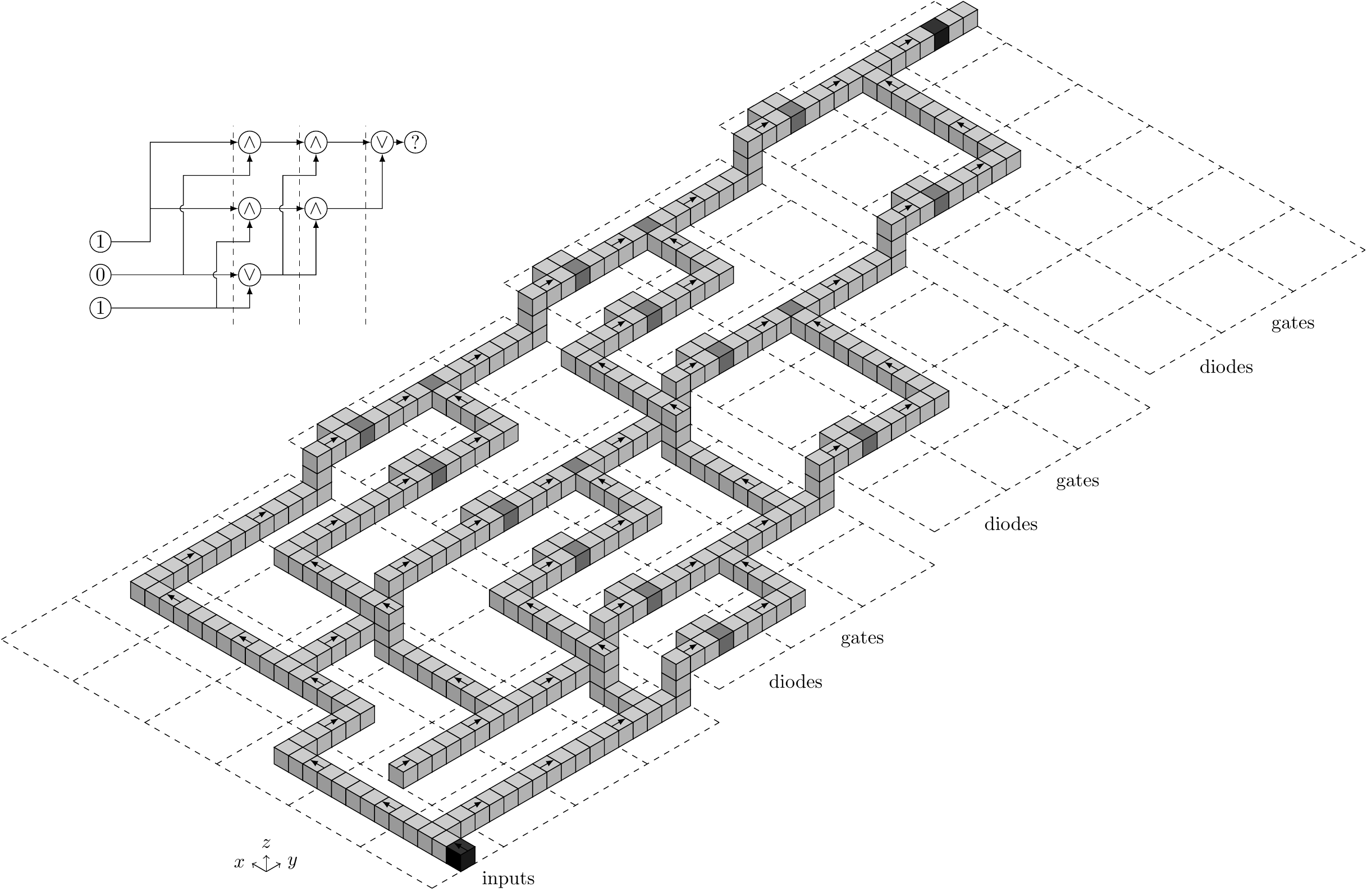}}
    \caption{von Neumann radius one in three dimensions simulating a layered
    circuit,
    with information flowing in three directions which will correspond to
    $x,y,z$. Light colored cubes contain $5$ grains ($\theta-1$), darker cubes
    contain $4$ grains ($\theta-2$), and black cubes are $\{0\}^d$ and the
    questioned cell (they both contain $5=\theta-1$ grains). The circuit
    is implemented with macrocells of size $5 \times
    5$, and a $z$-shift of $2$ units. Arrows indicate the direction of
    information flow.}
    \label{fig:3d}
  \end{figure}
  
  A sandpile model with such $x,y,z$ can simulate at a local level von Neumann
  of radius one in three dimensions (neighbours are given by $x,y,z$), itself
  simulating a circuit as in Figure~\ref{fig:3d}. Indeed, information flows in
  at most one direction ($x,y,z \in \neighborhood$ implies that information
  flows in the expected direction, and the reverse direction may or may not be
  in $\neighborhood$, as it is nor important nor an issue), and in other cases
  transmissions do not interfere one with another from the condition that
  $ax+by+cz \notin \neighborhood$.
  
  Hence taking cells from the three-dimensional grid generated by $\{x,y,z\}$,
  and placing:
  \begin{itemize}
    \item $\theta-\distribution(u)$ grain with $u\in \{x,y,z\}$ for cells with
      $5$ grains in von Neumann ($\theta-1$ in that model),
      depending on the expected in-neighbor
      $u$, or two in-neighbors and the minimum number of sand grains received
      from one of them,
    \item $\theta-\distribution(u)-\distribution(v)$ grains with $u,v \in
      \{x,y,z\}$ for cells with $4$ grains in von Neumann
      ($\theta-2$ in that model), depending on
      the two expected in-neighbors $u,v$,
    \item no grain for empy cells in von Neumann,
  \end{itemize}
  the sandpile model can simulate von Neumann radius one in three dimensions,
  itself simulating a circuit. The transformation is
  performed in constant parallel time, $\AC^0$.
\end{proof}

It follows that all $d$-dimensional sandpile models with $d \geq 3$
(verifying Equation~(\ref{eq:span}))
are $\Poly$-complete to predict.

\begin{corollary}
  \label{coro:3d}
  $\FSPRED$, $\SPRED$ and $\PRED$ are $\Poly$-complete for any sandpile model
  in dimension $d \geq 3$.
\end{corollary}

\begin{proof}
  Let $M=\structure{\neighborhood,\distribution,\theta}$ be a sandpile model in
  dimension $d \geq 3$, and let $\lVert u \rVert$ denote the Euclidean norm
  ($\ell_2$-norm) of cell $u \in \Z^d$. We define $x,y,z$ as follows.
  \begin{itemize}
    \item $x$ is a cell inside $\neighborhood$ of maximal norm, {\em i.e.}
      $$x \in \arg\max_{u \in \neighborhood} \{ \lVert u \rVert \}.$$
    \item $y$ is a cell inside $\neighborhood$ of maximal norm when projected
      onto the $(d-1)$-dimensional subspace orthogonal to the line defined by
      points $\{0\}^d,x$, {\em i.e.}
      $$y \in \arg\max_{u \in \neighborhood} \{ \lVert p_{x^\bot}(u) \rVert \}$$
      with $p_{x^\bot}$ the projection onto $\{ v \in \Z^d \mid v \cdot x = 0 \}$.
    \item $z$ is a cell inside $\neighborhood$ of maximal norm when projected
      onto the $(d-2)$-dimensional subspace orthogonal to the plane defined by
      points $\{0\}^d,x,y$, {\em i.e.}
      $$z \in \arg\max_{u \in \neighborhood} \{ \lVert p_{xy^\bot}(u) \rVert \}$$
      with $p_{xy^\bot}$ the projection onto $\{ v \in \Z^d \mid \forall a,b
      \in \Z: v \cdot (ax+by) = 0 \}$.
  \end{itemize}
  From Equation~(\ref{eq:span}) such $x,y,z$ exist and are non-colinear.

  Furthermore, it is always possible to choose $x,y,z$ such that $ax+by+xz
  \notin \neighborhood$ for $a,b,c \in \Z$ different than in the statement of
  Lemma~\ref{lemma:3d}. Indeed, a linear combination with two non-null
  components such that $ax+by+xz \in \neighborhood$ either gives another vector
  of maximal norm that may replace one of $y,z$ (if the projection onto $x^\bot$ or
  $xy^\bot$ is negative on all components, then it becomes positive on one
  component), or contradicts the maximality of $x$, $y$ or $z$ (if the
  projection onto $x^\bot$ or $xy^\bot$ is positive on one component).

  From Lemma~\ref{lemma:3d} we can conclude that $\FSPRED$ is $\Poly$-complete,
  and Proposition~\ref{prop:hierarchy} implies that $\SPRED$ and $\PRED$ are also
  $\Poly$-complete.
\end{proof}

Writting formal proofs of $\Poly$-completeness via reduction from some $\CVP$
problem requires a substantial amount of precisions, and some obvious details
are often not mentionned, though they may be key in other contexts (such as the
symmetry of von Neumann neighborhood and the dynamics of layers equipped with
diodes in the proof of Theorem \ref{theorem:VN3d}). In \cite{gmpt17} the
authors present a general framework to prove $\Poly$-completeness using Banks'
technic in cellular automata. A neat formalism is introduced, defining what is
meant by ``{\em simulating a gate set}''\footnote{in a nutshell: a macrocell
must be in some {\em valid} state, and depending on the state of its neighboring macrocells,
change to another {\em valid} state, within some common delay so that the
simulation remains synchronised. Valid states ensure that
nothing unexpected happens.}. Then results are presented, of
the form: if
a cellular automaton can simulate a set of gates $A$ (for example $A=\set {\text{{\em
and}, {\em or}, {\em cross-over}}}$), then its prediction problem is hard for
some class $C$ (in this example $C=\Poly$). A novel distinction appears to be
fundamental for the dynamical complexity of discrete dynamical systems: whether
is is possible to build {\em re-usable simulation} of wires and gates, or not
({\em weak simulation}). Indeed, re-usable simulation brings
$\Poly$-completeness from {\em and} and {\em or} gates only, because it is
possible to build a planar cross-over gadget with re-usable
simulation of monotones gates. This result is surprising
compared to the characterization of Boolean gates allowing planar
cross-over presented in~\cite{mccoll81}, which is not the case of
any set of monotone gates. It comes from the dynamical nature of
circuit simulation with discrete dynamical systems, and the
possibility to re-use wires that is not present in the original
circuit model.

In the context of sandpile models, circuit simulation in two dimensions is
harder to achieve precisely because of the difficulty to create cross-over (see
Section \ref{s:2d}). The planar cross-over gadget in re-usable simulation seems
however difficult to apply to sandpile models, since it exploits delays in
signal transmissions, which is in contradiction with the Abelian
Proposition~\ref{prop:abelian} telling that order of topplings do not matter.

\section{The two-dimensional case}
\label{s:2d}

No two-dimensional sandpile model is known to be efficiently predictable in
parallel, \ie such that its prediction problem is in $\NC$.
As we will see, some
slight extensions of von Neumann sandpile model of radius one turn out to have
$\Poly$-complete prediction problems, but the computational complexity of
predicting the original model of Bak, Tang and Wiesenfeld
\cite{1987-BakTangWiesenfeld-SOC} remains open. This is
considered as the major open problem regarding the complexity of prediction in
sandpile models. It is also open for Moore neighborhood of radius one.

\begin{open}
  Consider the von Neumann and Moore sandpile models in dimension two. Are $\PRED,
  \SPRED$ and $\FSPRED$ in $\NC$, $\Poly$-complete, or neither?
\end{open}

The third possibility comes from the fact that under the assumption
$\NC \neq \Poly$,
there exist problems in $\Poly$ that are neither in $\NC$ nor $\Poly$-complete
\cite{regan1997}. The difficulty in applying Banks' approach here is to
overcome planarity imposed by the two-dimensional grid with von Neumann or
Moore neighborhood, since the monotone planar circuit value problem ($\MPCVP$)
is in $\NC$. In fact, it has been proven to be impossible to perform {\em
elementary forms of signal cross-over} in this model
\cite{2006-Goles-CrossingInfo2DBTWSandpile}. The precise statement is a bit
technical to state\footnote{because an impossibility result
requires to define in full generality what is considered as a cross-over.}, but
corresponds neatly to the intuition of having two
potential sequences of cell topplings representing two wires that may convey a
bit of information, and cross each other without interacting ({\em i.e.} they
independently transport information). We shall call these {\em elementary forms
of signaling} since other ways to encode the transportation of information in
sandpile may be found, but as emphasized by Delorme and Mazoyer in
\cite{Delorme2001} to quantify over such possible encodings and give
fully general
impossibility results, is an issue.

\begin{theorem}[\cite{2006-Goles-CrossingInfo2DBTWSandpile}]
  It is impossible to perform {\em elementary forms of signal cross-over} in von
  Neumann and Moore neighborhoods of radius one.
\end{theorem}

However, as soon as the neighborhood is a bit extended, it turns out to be possible to
perform a cross-over, as expressed in the following result.

\begin{theorem}[\cite{2010-FormentiGolesMartin-KSPMAP,2006-Goles-CrossingInfo2DBTWSandpile}]
  In two dimensions, von Neumann neighborhood of radius $r \geq 2$ has
  $\Poly$-complete prediction problems $\FSPRED$, $\SPRED$ and $\PRED$. It is
  also the case for the non-deterministic Kadanoff sandpile model of radius $r
  \geq 2$.
\end{theorem}

In dimension two, the radius $r$ {\em Kadanoff} sandpile model
is defined in \cite{2010-FormentiGolesMartin-KSPMAP} as the
non-deterministic application of the
one-dimensional Kadanoff sandpile model (see Remark
\ref{remark:kadanoff}) in the two directions of the plane,
plus a monotonicity property that needs to be preserved, and
the prediction problem asks for the existence of an avalanche
reaching the questioned cell (in the circuit implementation
any cell topples at most once, hence it corresponds to $\FSPRED$).

Augmenting a little bit the neighborhood allows to perform cross-over and
simulate $\MCVP$ instances. More has been said in \cite{np18} on cross-over
impossibility, and as an immediate corollary we have that elementary forms of signal
cross-over are impossible when the graph supporting the dynamics is planar
(which is the case for von Neumann of radius one, but not of greater radii).

\begin{lemma}[\cite{np18}]
  If a sandpile model can implement an {\em elementary form of cross-over},
  then it can implement one such that the two {\em elementary signals} do not
  have any cell in common. This is true on any Eulerian digraph.
\end{lemma}

Also, the fact that a neighborhood can or cannot perform cross-over (given the
distribution $\distribution_1$ sending $1$ grain to each out-neighbor) is
intrinsically discrete. Indeed, if one thinks about the {\em shape} of some
neighborhood as a continuous two-dimensional region that we can scale and
place on a grid to get a (discrete) neighborhood, then impossibility to perform
cross-over cannot be characterized in terms of shape.

\begin{theorem}[\cite{np18}]
  Any shape can perform cross-over starting from some scaling ratio.
\end{theorem}

Even a circle, with a big enough scaling ratio, gives a neighborhood that can perform
cross-over and have a $\Poly$-complete prediction problem.

The precise conditions for $\Poly$-completeness of the prediction problems in
two-dimensions are still to be found. Having a precise characterization would
be of great interest to shed light on the universality of Banks' approach: if a
sandpile model has a $\Poly$-complete prediction problem, then for sure it can
simulate circuits. But are there ways of doing so with {\em non-elementary
forms of signaling}?

In \cite{1999-Moore-complexitySandpiles}, after recalling that $\MPCVP$
(monotone planar $\CVP$) is in $\NC$, Moore and Nilsson gave insights on the
possibilities to reduce circuit simulation to sandpile dynamics for von Neumann
radius one in two dimensions: ``{\em two-dimensional
sandpiles differ from planar Boolean circuits in two ways.
First, connections between sites are bidirectional. Secondly,
a given site can change a
polynomial number of times in the course of the sandpile's evolution. Thus
we really have a three-dimensional Boolean circuit of polynomial depth,
with layers corresponding to successive steps in the sandpile's space-time}.''

Let us conclude this section with a lower bound on the complexity of
two-dimensional sandpile models. As noticed by Miltersen in \cite{Miltersen2005}, a
corollary of Theorem~\ref{theorem:VN3d} is that the
two-dimensional von Neumann sandpile model of radius one is
$\NC^1$-hard since it can simulate monotone planar circuits (an $\MPCVP$
instance can be reduced to a $\FSPRED$ instance with an $\AC^0$ algorithm,
without implementing cross-over), and $\MPCVP$ is $\NC^1$-hard since evaluating
a Boolean formula (for which the circuit is a tree) is $\NC^1$-complete
\cite{buss87}. From Lemma~\ref{lemma:3d} and Corollary~\ref{coro:3d} this
observation extends to any two-dimensional sandpile model, since the third
dimension is used exclusively for cross-over gates.

\begin{theorem}
  For any two-dimensional sandpile model, $\FSPRED$, $\SPRED$ and $\PRED$ are
  $\NC^1$-hard.
\end{theorem}

\section{Simulations between sandpile models}

As seen in the previous sections the prediction problems for sandpile models are dimension sensitive 
but once the dimension is fixed they turn out to be closely related to one another. This section provides
a notion of simulation between sandpile models and show that simulating a run of a sandpile model
by another with different parameters (neighborhood or distribution) has a polynomial cost.
Only sequential update policy is considered in this section. Moreover, all finite configurations mentioned are intended to
belong to the elementary hypercube. For this reason those hypothesis is omitted from the statements.
\medskip

Given a configuration $c\in\cfgs$, a \emph{firing sequence} for $c$ is a sequence $(x_1,x_2,\ldots,x_n)\in\Z^d$
of cells such that there exist configurations $c_1=c, c_2, \ldots, c_n$ such that $c_i\donnes{x_i} c_{i+1}$ for all $i\in\set{1, \ldots, n-1}$ and
\begin{enumerate}
\item $c_1(x_1)=\theta-1$;
\item $x_{i+1}-x_i\in\neighborhood$ for all $i\in\set{1,2,\ldots,n-1}$;
\item $c_i(x_i)=\theta$ for $i\in\set{2, \ldots, n}$.
\end{enumerate}
Remark that when a firing sequence has been triggered \ie when a sand grain has been added at $x_1$ in $c_1$, grains may be distributed on set of cells which are in the neighborhood of some element
of the sequence. The \emph{hitting set} collects precisely this information. More formally, the \emph{hitting set} induced by a
firing sequence $(x_1, x_2, \ldots, x_n)\in\Z^d$ for a finite configuration $c$ is the set of pairs $(x,u)$ where $x\in\Z^d$ is a cell and $u\in\N\setminus\set{0}$ is the total number of grains that $x$ has received
when all the cells of the firing sequence have been fired. 
An $\neighborhood$-path from the cell $x\in\Z^d$ to 
$y\in\Z^d$ is an sequence $(z_1, z_2, \ldots, z_n)$ of cells
such that $z_1=x$, $z_n=y$ and $z_{i+1}-z_i\in\neighborhood$ for all $i\in\set{1,\ldots,n-1}$. Remark that if $\neighborhood$ is complete, then
there always exists an $\neighborhood$-path between any pair of cells. Given a finite configuration $c$, 
let $V(c)\subset\Z^d$ be the convex hull of points $x\in\Z^d$
such that $c_x\ne0$.
A \emph{detector} cell $y$ for a finite configuration $c$ 
with hitting set $H$ is a cell such that there exists an integer $u>0$
such that $(y,u)\in H$ and $y\notin V(c)$.
In other words, a detector cell for a finite configuration $c$ is a cell which is ``outside''
$c$ and which may receive a sand grain if some cell $z$ ``inside'' $c$ is triggered. Of course, if given a finite configuration $c$, for all $x\in\Z^d$ we have that  $c+\indic_x$ is stable, then
$c$ has no detector cells. 
Figure~\ref{fig:simulation-notions} illustrates
all these recent notions.


\begin{figure}[htb]
\newcounter{gridcounterx}
\newcounter{gridcountery}
\begin{center}
\begin{tikzpicture}[scale=.4,font=\footnotesize]
\begin{scope}[xshift=3.5cm,yshift=-8cm]
  \draw[black!20] (-1,-1) grid (3,3);
  \filldraw[fill=black!30] (0,0) rectangle ++ (1,1);
  \foreach \x/\y/\v in {0/2/1,0/1/1,2/0/2,-1/-1/1,0/-1/1}
    \draw (\x,\y) rectangle node{\v} ++ (1,1);
  \node[rectangle,align=center](z) at (-2,2) {cell at\\$(0,0)$};
  \draw[-latex] (z) -- (-2,.5) -- (.5,.5);
\end{scope}
\begin{scope} 
\filldraw[black!30,shift={(-.5,-.5)}] (3,2) rectangle ++ (1,1);
  \draw[shift={(-.5,-.5)}] (-.5,-.5) grid (8.5,6.5);
  \setcounter{gridcounterx}{0} 
  \setcounter{gridcountery}{0} 
  \foreach \v in {
    0,0,0,0,0,0,0,0,
    0,1,2,4,0,3,1,0,
    0,0,1,6,1,4,2,0,
    0,2,5,0,3,2,0,0,
    0,1,3,5,5,3,1,0,
    0,0,0,0,0,0,0,0
  }{
    \node at (\value{gridcounterx},\value{gridcountery}) {$\v$};
    \addtocounter{gridcounterx}{1}
    \ifthenelse{\value{gridcounterx}>7} 
    {
      \setcounter{gridcounterx}{0} 
      \addtocounter{gridcountery}{1}
    }{}
  }
\end{scope}
\begin{scope}[xshift=9cm,yshift=3cm]
\node at (0,0) {$\overset{F}{\mapsto}$};
\end{scope}
\begin{scope}[xshift=11cm]
\filldraw[black!30,shift={(-.5,-.5)}] (5,2) rectangle ++ (1,1);
  \draw[shift={(-.5,-.5)}] (-.5,-.5) grid (8.5,6.5);
  \setcounter{gridcounterx}{0} 
  \setcounter{gridcountery}{0} 
  \foreach \v in {
    0,0,0,0,0,0,0,0,
    0,1,3,5,0,3,1,0,
    0,0,1,0,1,6,2,0,
    0,2,5,1,3,2,0,0,
    0,1,3,6,5,3,1,0,
    0,0,0,0,0,0,0,0
  }{
    \node at (\value{gridcounterx},\value{gridcountery}) {$\v$};
    \addtocounter{gridcounterx}{1}
    \ifthenelse{\value{gridcounterx}>7} 
    {
      \setcounter{gridcounterx}{0} 
      \addtocounter{gridcountery}{1}
    }{}
  }
\end{scope}
\begin{scope}[xshift=20cm,yshift=3cm]
\node at (0,0) {$\overset{F}{\mapsto}$};
\end{scope}
\begin{scope}[xshift=22cm] 
\filldraw[black!30,shift={(-.5,-.5)}] (3	,4) rectangle ++ (1,1);
  \draw[shift={(-.5,-.5)}] (-.5,-.5) grid (8.5,6.5);
  \setcounter{gridcounterx}{0} 
  \setcounter{gridcountery}{0} 
  \foreach \v in {
    0,0,0,0,0,0,0,0,
    0,1,3,5,1,4,1,0,
    0,0,1,0,1,0,2,2,
    0,2,5,1,3,3,0,0,
    0,1,3,6,5,4,1,0,
    0,0,0,0,0,0,0,0
  }{
    \node at (\value{gridcounterx},\value{gridcountery}) {$\v$};
    \addtocounter{gridcounterx}{1}
    \ifthenelse{\value{gridcounterx}>7} 
    {
      \setcounter{gridcounterx}{0} 
      \addtocounter{gridcountery}{1}
    }{}
  }
\end{scope}
\begin{scope}[xshift=20cm,yshift=3cm]
\node at (0,0) {$\overset{F}{\mapsto}$};
\end{scope}
\begin{scope}[xshift=25.5cm,yshift=-2cm]
\node at (0,0) {\rotatebox{-90}{$\overset{F}{\mapsto}$}};
\end{scope}
\begin{scope}[xshift=22cm,yshift=-9.5cm]
\filldraw[black!30,shift={(-.5,-.5)}] (2,3) rectangle ++ (1,1);
  \draw[shift={(-.5,-.5)}] (-.5,-.5) grid (8.5,6.5);
  \setcounter{gridcounterx}{0} 
  \setcounter{gridcountery}{0} 
  \foreach \v in {
    0,0,0,0,0,0,0,0,
    0,1,3,5,1,4,1,0,
    0,0,1,0,1,0,2,2,
    0,2,6,2,3,3,0,0,
    0,1,3,0,5,4,1,0,
    0,0,0,1,0,0,0,0,
    0,0,0,1,0,0,0,0
  }{
    \node at (\value{gridcounterx},\value{gridcountery}) {$\v$};
    \addtocounter{gridcounterx}{1}
    \ifthenelse{\value{gridcounterx}>7} 
    {
      \setcounter{gridcounterx}{0} 
      \addtocounter{gridcountery}{1}
    }{}
  }
\end{scope}
\begin{scope}[xshift=20cm,yshift=-6cm]
\node at (0,0) {$\overset{\,\,F}{\rotatebox{-180}{$\mapsto$}}$};
\end{scope}
\begin{scope}[xshift=11cm,yshift=-9.5cm]
\filldraw[black!15,shift={(-.5,-.5)}] (7,2) rectangle ++ (1,1);
\filldraw[black!15,shift={(-.5,-.5)}] (3,6) rectangle ++ (1,1);
\filldraw[black!15,shift={(-.5,-.5)}] (3,5) rectangle ++ (1,1);
\filldraw[black!15,shift={(-.5,-.5)}] (2,5) rectangle ++ (1,1);
  \draw[shift={(-.5,-.5)}] (-.5,-.5) grid (8.5,6.5);
  \setcounter{gridcounterx}{0} 
  \setcounter{gridcountery}{0} 
  \foreach \v in {
    0,0,0,0,0,0,0,0,
    0,1,\mathbf{3},\mathbf{5},\mathbf{1},\mathbf{4},1,0,
    0,\mathbf{1},\mathbf{2},0,1,0,2,\mathbf{2},
    0,2,0,\mathbf{2},\mathbf{5},\mathbf{3},0,0,
    0,1,\mathbf{4},0,5,\mathbf{4},1,0,
    0,0,\mathbf{1},\mathbf{1},0,0,0,0,
    0,0,0,\mathbf{1},0,0,0,0
  }{
    \node at (\value{gridcounterx},\value{gridcountery}) {$\v$};
    \addtocounter{gridcounterx}{1}
    \ifthenelse{\value{gridcounterx}>7} 
    {
      \setcounter{gridcounterx}{0} 
      \addtocounter{gridcountery}{1}
    }{}
  }
\end{scope}
\end{tikzpicture}
\end{center}
\caption{Sandpile model with $\theta=6$, neighborhood
and distribution function as depicted on bottom left. The evolution
starts from the finite configuration on top left (cells not drawn are supposed to contain $0$). Dark-grayed cells (taken in the order indicated by $\mapsto$) are the firing sequence. In the final stable configuration, cells in the hitting
set have their content in bold face, while detector cells have
a light-gray background.}
\label{fig:simulation-notions}
\end{figure}
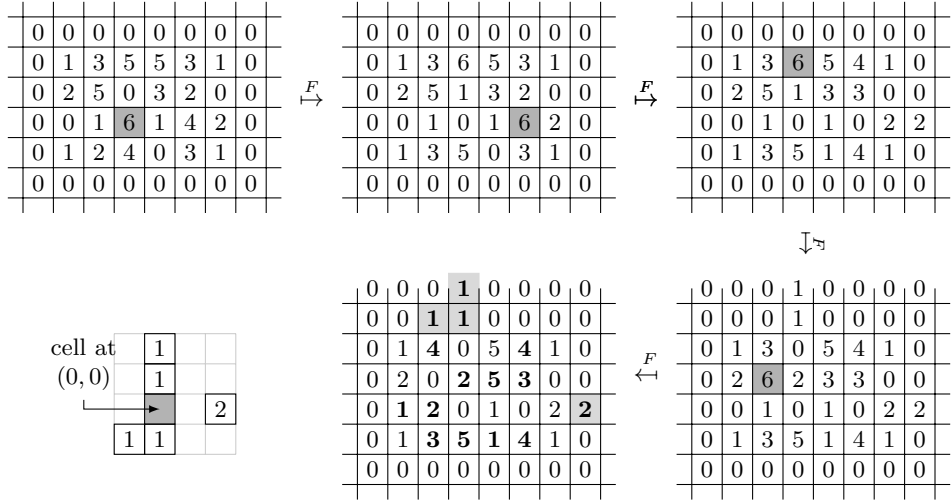

Given two sandpile models $M_1\equiv\structure{\neighborhood_1,\distribution_1,\theta_1}$
and $M_2\equiv\structure{\neighborhood_2,\distribution_2,\theta_2}$, we say that $M_2$
\emph{simulates} $M_1$ if 
there exist a computable transformation $h:\cfgs\to\cfgs$ such that
for all finite configurations $c,c'\in\cfgs$ and all pair of cells $x,y\in\Z^d$ such
that
\begin{enumerate}
\item $c,c'$ are stable for $M_1$;
\item $c+\indic_x\sdonnes c'$ according to $M_1$;
\item $y$ is a detector cell for $c$;
\end{enumerate}
 and there exists a finite configuration $c''$ such that
\begin{enumerate}
\item $h(c),c''$ are stable for $M_2$;
\item $h(c)+\indic_x\sdonnes c''$ according to $M_2$;
\item if $c'(y)>0$ then $c''(y)>0$.
\end{enumerate}
In other words, $M_2$ simulates $M_1$ if starting on a computable encoding of an initial configuration $c$ of $M_1$
the stable configuration which is reached afterwards contains the same bit of information in the detector cell $y$ as
$M_1$.

\begin{lemma}
Consider two complete neighborhoods $\neighborhood_1, \neighborhood_2\subset\Z^d$ 
such that $\neighborhood_1=\neighborhood_2\cup\set{u}$ for $u\in\Z^d$.
For any sandpile model $M_2\equiv\structure{\neighborhood_2,\distribution_2,\theta_2}$, 
there exists a sandpile  model $M_1\equiv\structure{\neighborhood_1,\distribution_1,\theta_1}$ 
which simulates $M_2$.
\end{lemma}
\begin{proof}
Let $c\in\cfgs$ 
be a stable configuration for $M_2$ and choose $x\in\Z^d$. Let $c'$ be such that $c+\indic_x\sdonnes c'$
according to $M_2$ and let $y$ be a detector cell for $c'$. Let $F_2$ be the firing
sequence $F_2\equiv(x_1, \ldots, x_n)$ associated with $c+\indic_x$ and let $H_2$ be the induced hitting set. 
The idea is to take the same firing sequence also for the model
$M_1$ that we are going to define but to craft the configuration $c'$ realizing the circuit for $M_1$ in such a way to prevent
unnecessary supplementary firings or before-time firings because of grains dropped forward by the $u$ component of the
neighborhood. Hence, let $F_1=F_2$ be the firing sequence for the model $M_1$.
Define $\theta_1=\theta_2+1$ and
\[
\forall v\in\neighborhood_1,\;\distribution_1(v)=
\begin{cases}
\distribution_2(v)&\text{if }v\in\neighborhood_2\\
1&\text{otherwise.}
\end{cases}
\]
Let us build a new finite configuration $c''$ as follows
\[
\forall z\in\Z^d,\;c''(z)=
\begin{cases}
\theta_1-1&\text{if }z=x\\
\theta_1-v&\text{if } z\in F_1\text{ and }(z,v)\in H_2\\
0&\text{otherwise.}
\end{cases}
\]
Let $c'''$ be the stable configuration such that $c''\sdonnes c'''$ according to $M_1$.
It is clear that $c'(y)>0$ implies $c'''(y)>0$. Indeed, if $c'(y)>0$,
then there exists $x_i\in F_1=F_2$ such that 
$y-x_i\in\neighborhood_2\subset\neighborhood_1$. Hence,
$c'''(y)>0$.
\end{proof}

The proofs of the following lemmas are very similar to the ones of previous lemmas and thus they are omitted.

\begin{lemma}
Consider a neighborhood $\neighborhood\subseteq\Z^d$ (not necessarily complete)
and a sandpile model $M_1\equiv\structure{\neighborhood,\distribution_1,\theta_1}$.
Define the model $M_2\equiv\structure{\neighborhood,\distribution_2,\theta_1+k}$
for some $k\in\N_+$ and such that
\begin{enumerate}
\item there exists a unique $u\in\neighborhood$ for which $\distribution_2(u)=\distribution_1(u)+k$; 
\item $\forall v\in\neighborhood,\;(u\ne v) \Rightarrow \distribution_2(v)=\distribution_1(v)$.
\end{enumerate}
Then, $M_2$ simulates $M_1$.
\end{lemma}

\begin{lemma}
Consider a neighborhood $\neighborhood\subseteq\Z^d$ (not necessarily complete)
and a sandpile model $M_1\equiv\structure{\neighborhood,\distribution_1,\theta_1}$.
Define the model $M_2\equiv\structure{\neighborhood,\distribution_2,\theta_2}$
such that
\begin{enumerate}
\item there exists a unique $u\in\neighborhood$ for which $\distribution_2(u)=\distribution_1(u)-k>0$ 
for some $k\in\N_+$; 
\item $\forall v\in\neighborhood,\;(u\ne v) \Rightarrow \distribution_2(v)=\distribution_1(v)$;
\item $\theta_2=\theta_1-k$.
\end{enumerate}
Then, $M_2$ simulates $M_1$.
\end{lemma}

Remark that from a computational complexity point of view, the constructions of the previous
lemma come at no cost. However, each single simulation has a cost which essentially consists 
in computing the firing sequence for the original system (the hitting set can be computed at a 
constant multiplicative cost while computing the firing sequence). By Theorem~\ref{theorem:poly},
this can be done in polynomial time (using a stack for example).

The notion of simulation between sandpile models induces a preorder structure on the set of
sandpile models in the same dimension. It is an interesting research direction to explore the
properties of such an order and see if and under which form there exists a notion of
universality.

\section{Undecidability on infinite configurations}
\label{s:undecidable}

A further generalisation of sandpile models consists in relaxing the finiteness
of the initial configuration. As noted in~\cite{cairns2018}, allowing any
configuration initially written on the tape of the Turing machine would make
the prediction problem trivially undecidable in any dimension, and restricting to periodic
configurations comes down to studying a finite region on a torus and is
therefore decidable in any dimension (since given a fixed number of sand grains
there are finitely many configurations). On the other hand,
when the initial configuration given as input to
the prediction problem is \emph{ultimately periodic} (\ie periodic except on a finite 
region), then the following result holds.

\begin{theorem}[\cite{cairns2018}]
  $\PRED$ extended to ultimately periodic configurations (the input consists in the finite non-periodic region, plus the finite periodic pattern repeated all around) in dimension three is undecidable.
\end{theorem}

Finally, if we further relax the number-conservation property and allows the distribution
functions to ``eat'' grains, then one obtains \emph{sand automata}~\cite{formenti2003}.
Without going into the details of their precise definition (which is a bit involved), we can
just recall that they are a special type of cellular automata particularly adapted to the sandpile
``playground''~\cite{dennunzio2009,dennunzio2008}. The following result is interesting 
in our context.

\begin{theorem}[\cite{formenti2005,formenti2007}]
Ultimate (temporal) periodicity is undecidable for sand automata on finite configurations.
\end{theorem}

This last result tells that (somewhat unsurprisingly) sand automata are highly unpredictable.

\section*{Acknowledgments}

The authors thank
the \emph{Young Researcher} project ANR-18-CE40-0002-01 ``FANs'',
the project ECOS-CONICYT C16E01,
the project STIC AmSud CoDANet 19-STIC-03 (Campus France 43478PD).

\bibliographystyle{plain}
\bibliography{biblio}

\end{document}